\documentclass[authoryear]{imsart}

\usepackage{etex}
\reserveinserts{18}

\usepackage{amsmath,amssymb,amsfonts,amsthm}
\usepackage{amsxtra,amstext,latexsym,dsfont} 
\usepackage{mathtools}
\usepackage{wasysym}
\usepackage{stmaryrd}
\usepackage{yfonts}

\usepackage{parallel,accents}

\normalfont
\usepackage{bm}

\usepackage[comma]{natbib}

\usepackage[pdftex]{graphicx}
\usepackage[dvipsnames]{xcolor}

\usepackage{pgf,tikz}
\usetikzlibrary{shapes}
\usetikzlibrary{snakes}
\usetikzlibrary{backgrounds}
\usetikzlibrary{arrows}
\usetikzlibrary{patterns}
\usetikzlibrary{fit}
\usetikzlibrary{calc}
\usetikzlibrary{matrix}

\usepackage{lscape} 

\usepackage{array,longtable,booktabs,float,rotating}
\usepackage{dcolumn,ctable}

\usepackage[labelfont={bf},ruled,labelsep=period,small]{caption}

\usepackage[flushleft]{threeparttable}

\usepackage[linkcolor=black,colorlinks=true,citecolor=black,
bookmarks=true,bookmarksopen=true,bookmarksnumbered=true,bookmarksopenlevel=0,%
hyperindex=true,pdfpagemode=UseOutlines,linktocpage=true,pdfpagelayout=TwoColumnLeft,
pdftitle=Biostatistics,pdfstartview=FitH,pdffitwindow=true,urlcolor=blue]{hyperref}

\usepackage{ragged2e}

\usepackage{multirow,multicol}

\usepackage{fancyhdr,extramarks}

 \makeatletter
 \newcommand\sub{\@startsection%
     {subsubsection}{5}{0mm}{-1\baselineskip}{.01\baselineskip}%
     {\normalfont\itshape}}
 \renewcommand\subsubsection{\@startsection%
     {subsubsection}{3}{0mm}{-1\baselineskip}{.01\baselineskip}%
     {\normalfont\itshape}}
 \makeatother

\makeatletter
        \newcommand\Appendix[2][?]{%
            \refstepcounter{section}%
            \addcontentsline{toc}{appendix}%
                {\protect\numberline{\appendixname~\thesection}#1}%
            {\raggedleft\bfseries \appendixname\
                \thesection\par \centering#2\par}%
                \sectionmark{#1}%
                \@afterheading
                \addvspace{\baselineskip}}
        \newcommand\sAppendix[1]{%
            \raggedleft\bfseries\appendixname\par
            \@afterheading\addvspace{\baselineskip}}
    \makeatother

\normalsize \makeindex

\newcolumntype{A}{>{\centering}p{100pt}}
\newlength\savedwidth

\def\coldot{.}%
{\catcode`\.=\active%
    \gdef.{$\egroup\setbox2=\hbox to \dimen0 \bgroup$\coldot}}
\def\rightdots#1{%
    \setbox0=\hbox{$1$}\dimen0=#1\wd0%
    \setbox0=\hbox{$\coldot$}\advance\dimen0 \wd0%
    \setbox2=\hbox to \dimen0 {}%
    \setbox0=\hbox\bgroup\mathcode`\.="8000 $}
\def\endrightdots{$\hfil\egroup\box0\box2}

\newcolumntype{d}[1]{D{.}{.}{#1}}
\newcolumntype{A}{>{\centering}p{100pt}}
\newcolumntype{.}{D{.}{.}{-1}}
\newcolumntype{P}[2]{>{#1\raggedright\arraybackslash}p{#2}}

\DeclareFontFamily{U}{euc}{}
\DeclareFontShape{U}{euc}{m}{n}{<-6>eurm5<6-8>eurm7<8->eurm10}{}%

\theoremstyle{plain}      
\theoremstyle{plain}      
\theoremstyle{plain}      
\theoremstyle{plain}      
\theoremstyle{definition} 
\theoremstyle{definition} 
\theoremstyle{definition} 
\theoremstyle{plain} \newtheorem{cor}{Corollary}
\theoremstyle{definition} 
\theoremstyle{plain} \newtheorem{pro}{Proposition}
\theoremstyle{definition} 
\theoremstyle{definition} 
\theoremstyle{definition} 

\newcounter{nctr}
\newenvironment{3table}{\begin{threeparttable}}{\end{threeparttable}}

\newenvironment{en}{\begin{enumerate}}{\end{enumerate}}

\usepackage[rightbars,color]{changebar}
\cbcolor{blue}

\usepackage{url}
\usepackage{fancyvrb}

\newcommand\mrow{\multirow}
\newcommand\mcol{\multicolumn}

\newcommand\bb{\mathbb}

\newcommand\te{\text}
\newcommand\ma[1]{\te{\bf{#1}}}
\newcommand\ca{\mathcal}

\newcommand\op{\operatorname}

\newcommand\as{^\ast}

\newcommand\argmin{\operatornamewithlimits{argmin}}

\newcommand\bias{\operatorname{\bb{B}ias}}
\newcommand\cov{\operatorname{\bb{C}ov}}

\newcommand\E{\bb{E}}

\newcommand\for{\,\,\forall\,\,}

\newcommand\iid{\op{iid}}

\newcommand\mse{\op{MSE}}

\newcommand\lb{\lbrace}
\newcommand\lt{\left}
\newcommand\lan{\langle}

\newcommand\pri{^\prime}

\newcommand{\plim}{\operatornamewithlimits{plim}}

\newcommand\rank{\te{rank}}

\newcommand\rb{\rbrace}
\newcommand\rt{\right}
\newcommand\ran{\rangle}

\newcommand\stack{\stackrel} 

\newcommand\tr{\op{tr}}

\newcommand\tth{^\text{th}}
\newcommand\var{\operatorname{\bb{V}ar}}

\newcommand\wh{\widehat}
\newcommand\wti{\widetilde}

\newcommand\R{\bb{R}}  

\newcommand\br{\ma{r}} 
\newcommand\bx{\ma{x}}
\newcommand\by{\ma{y}}
\newcommand\bz{\ma{z}}

\newcommand\bA{\ma{A}} 
\newcommand\bD{\ma{D}} 
\newcommand\bH{\ma{H}} %
\newcommand\bR{\ma{R}} 
\newcommand\bX{\ma{X}}

\newcommand\bZ{\ma{Z}}

\newcommand\bzero{\bm{0}} 

\newcommand\cI{\ca{I}} 
\newcommand\al{\alpha}
\newcommand\be{\beta}
\newcommand\ga{\gamma}
\newcommand\de{\delta}
\newcommand\ep{\varepsilon}

\newcommand\la{\lambda}
\newcommand\sig{\sigma}

\newcommand\bbe{\bm\beta}

\newcommand\bde{\bm\delta}
\newcommand\bep{\bm\varepsilon}

\newcommand\bla{\bm\lambda}

\newcommand\bGa{\bm\Gamma}

\usepackage{etoolbox}

\makeatletter
\providerobustcmd*{\bigcapdot}{%
  \mathop{%
    \mathpalette\bigop@dot\bigcap
  }%
}
\newrobustcmd*{\bigop@dot}[2]{%
  \setbox0=\hbox{$\bigcup\m@th#1#2$}%
  \vbox{%
    \lineskiplimit=\maxdimen
    \lineskip=-0.6\dimexpr\ht0+\dp0\relax
    \ialign{%
      \hfil##\hfil\cr
      $\m@th\centerdot$\cr
      \box0\cr
    }%
  }%
}
\makeatother

\makeatletter
\providerobustcmd*{\cupdot}{%
  \mathop{%
    \mathpalette\op@dot\cup
  }%
}
\providerobustcmd*{\capdot}{%
  \mathop{%
    \mathpalette\op@dot\cap
  }%
}
\newrobustcmd*{\op@dot}[2]{%
  \setbox0=\hbox{$\m@th#1#2$}%
  \vbox{%
    \lineskiplimit=\maxdimen
    \lineskip=-0.9\dimexpr\ht0+\dp0\relax
    \ialign{%
      \hfil##\hfil\cr
      $\m@th\cdot$\cr
      \box0\cr
    }%
  }%
}
\makeatother

\makeatletter
\providerobustcmd*{\subsetdot}{%
  \mathop{%
    \mathpalette\subop@dot\subset
  }%
}
\providerobustcmd*{\supsetdot}{%
  \mathop{%
    \mathpalette\subop@dot\supset
  }%
}
\newrobustcmd*{\subop@dot}[2]{%
  \setbox0=\hbox{$\m@th#1#2$}%
  \vbox{%
    \lineskiplimit=\maxdimen
    \lineskip=-0.9\dimexpr\ht0+\dp0\relax
    \ialign{%
      \hfil##\hfil\cr
      $\m@th\cdot$\cr
      \box0\cr
    }%
  }%
}
\makeatother

\makeatletter
\providerobustcmd*{\subseteqdot}{%
  \mathop{%
    \mathpalette\subeqop@dot\subseteq
  }%
}
\providerobustcmd*{\supseteqdot}{%
  \mathop{%
    \mathpalette\subeqop@dot\supseteq
  }%
}
\newrobustcmd*{\subeqop@dot}[2]{%
  \setbox0=\hbox{$\m@th#1#2$}%
  \vbox{%
    \lineskiplimit=\maxdimen
    \lineskip=-0.7\dimexpr\ht0+\dp0\relax
    \ialign{%
      \hfil##\hfil\cr
      $\m@th\cdot$\cr
      \box0\cr
    }%
  }%
}
\makeatother

\newcommand*\Otri{\ensuremath{
        \;
        \begin{tikzpicture}
          \draw[thick] (0pt,0pt) circle (6.5pt);
          \draw[thick] (-6pt,+0pt) -- (+5.50pt,-2.75pt);
          \draw[thick] (-6pt,+0pt) -- (+5.50pt,+2.75pt);
        \end{tikzpicture}
        \;
        }
}

\newcommand{\dobigtri}[1]{%
  \vcenter{#1\kern.2ex\hbox{$\Otri$}\kern.2ex}}

\newcommand*\oline[1]{%
  \vbox{%
    \hrule height 0.5pt
    \kern0.5ex
    \hbox{%
      \kern-0.1em
      \ifmmode#1\else\ensuremath{#1}\fi
      \kern-0.1em
    }
  }
}

\startlocaldefs

\newcommand{\cse}{\op{CSE}}

\newcommand{\Cor}{\op{Cor}}
\endlocaldefs

\begin{document}
\sloppy
\begin{frontmatter}
\title{Convex Combination \\ 
       of Ordinary Least Squares and Two-stage Least Squares Estimators}
\runtitle{Convex Combination of OLS and TSLS estimators}
\author{\fnms{Cedric E.~Ginestet}$^{1}$\ead[label=e1]{cedric.ginestet@kcl.ac.uk}\thanksref{t1},}
\author{\fnms{Richard Emsley}$^{2}$, and}
\author{\fnms{Sabine Landau}$^{1}$}
\address{$^{1}$ Biostatistics Department, Institute of Psychiatry, \\ 
        Psychology and Neuroscience, King's College London}
\address{$^{2}$ Centre for Biostatistics, Institute of Population Health, University of Manchester}
\runauthor{Cedric E. Ginestet}
\thankstext{t1}{This work was supported by an 
MRC project grant MR/K006185/1, Landau et al.~(2013-2016) entitled
``Developing methods for understanding mechanism in complex
interventions.'' We also would like to thank Stephen Burgess, Paul
Clarke, Graham Dunn, Andrew Pickles, and Ian White for useful
suggestions and discussions.}
\begin{abstract}
In the presence of confounders, the ordinary least squares (OLS) estimator is known to be biased. This problem can be remedied by using the two-stage least squares (TSLS) estimator, based on the availability of valid instrumental variables (IVs). This reduction in bias, however, is offset by an increase in variance. Under standard assumptions, the OLS has indeed a larger bias than the TSLS estimator; and moreover, one can prove that the sample variance of the OLS estimator is no greater than the one of the TSLS.  Therefore, it is natural to ask whether one could combine the desirable properties of the OLS and TSLS estimators. Such a trade-off can be achieved through a convex combination of these two estimators, thereby producing our proposed convex least squares (CLS) estimator. The relative contribution of the OLS and TSLS estimators is here chosen to minimize a sample estimate of the mean squared error (MSE) of their convex combination.  This proportion parameter is proved to be unique, whenever the OLS and TSLS differ in MSEs. Remarkably, we show that this proportion parameter can be estimated from the data, and that the resulting CLS estimator is consistent.  We also show how the CLS framework can incorporate other asymptotically unbiased estimators, such as the jackknife IV estimator (JIVE). The finite-sample properties of the CLS estimator are investigated using Monte Carlo simulations, in which we independently vary the amount of confounding and the strength of the instrument. Overall, the CLS estimator is found to outperform the TSLS estimator in terms of MSE. The method is also applied to a classic data set from econometrics, which models the financial return to education.
\end{abstract}
\arxiv{0000.0000}
\end{frontmatter}


\section{Introduction}\label{sec:intro}
Instrumental variables (IVs) estimation is one of the cornerstones of
modern econometric theory. The use of IVs has been described as ``only
second to ordinary least squares (OLS) in terms of methods used in
empirical economic research'' \citep[][p.89]{Wooldridge2002}. This
ranking of estimation techniques 
naturally leads to the following methodological questions: When should we prefer
IV estimation over OLS? Is it always preferable to use an instrument
even though this may substantially increase the variance of the
resulting estimator? 

In fields including econometrics and the social sciences, and
in some medical disciplines such as psychiatry, 
the direct randomized allocation of subjects to different experimental
conditions is rarely possible, thereby preventing such scientists from
inferring causal relations. Without adequate experimental manipulation,
the model's predictors may be correlated with the errors. When this is
case, we say that the predictors are \textit{endogenous}. The absence
of experimental manipulation in observational data, however, can be
addressed by using IVs to predict the alleged causal
variables. In particular, the resulting IV estimators allow to reduce
the bias of the estimated effect. The main difficulty in conducting
such IV analyses lies in the choice of appropriate \textit{exogenous}
instruments. Indeed, instruments are assumed to be solely correlated with the
outcome variable through the predictor. This specific assumption is
sometimes referred to as the \textit{exclusion criterion}, since it
disallows any direct effect of the instrument on the outcome. 

The first published use of IVs is commonly attributed to
\citet{Wright1928} in the context of microeconometrics, albeit this
has been historically disputed \citep[see][]{Stock2003}. This
estimation technique has been widely adopted in econometrics, and in
other social sciences, including psychology, epidemiology, public
health and political science. In particular, the use of IV methods has
now become an integral part of causal inference
\citep{Pearl2009}. The use of IVs in regression has been extended in
several directions, allowing two-sample estimation, for instance
\citep{Inoue2010}, and the selection of instruments using penalized
methods such as the LASSO \citep{Ng2009,Belloni2012}. 
More recently, these methods have become especially popular 
in the study of genetic variants, thereby demonstrating the wide
applicability of IV-based methods \citep{Palmer2012,Pierce2013}. The reader may
consult \citet{Wooldridge2002} and \citet{Cameron2005}
for an introduction to the use of instrumental variables in the context
of econometrics. A review of the assumptions underlying the use of IVs
is provided by \citet{Angrist2001}, and \citet{Heckman1997}; whereas
the finite-sample properties of IV estimators have been described by
\citet{Maddala1992} and \citet{Nelson1988}. 

While the asymptotic properties of IV estimators such as the two-stage
least squares (TSLS) are well-understood \citep{Staiger1997,Hahn2004};
in practice, it is not always clear whether or not using an IV
estimator over a simpler OLS estimator is necessarily
beneficial. Intuitively, since every IV is a
random variable, its inclusion in the analysis tends to increase the variance
of the resulting estimator. The magnitude of this increase in variance
is proportional to the correlation of the instrument with the
predictor. Poor or \textit{weak} instruments are variables that are
weakly correlated with the endogenous variables in the model. Thus,
although the use of an IV estimator is likely to lead to a significant
decrease in the bias of the OLS estimator, it will also yield a
more variable estimator. Since the true value of the parameters of
interest is unknown in practice, it is generally not possible to
evaluate whether the benefit of using a given set of instruments
outweighs the cost in variance of incorporating them into the
analysis. In addition, the use of weak instruments can also lead
to a substantial amount of finite-sample bias. Indeed, the use of weak
instruments has been studied by
\citet{Bound1995}, and these authors have shown that the inclusion of
instruments with only weak linear relationships with the endogenous
variables, tends to inflate the bias of the IV estimator; ultimately yielding
an estimator as biased as the original OLS estimator. 

In this paper, we address this issue by proposing a sample estimate
of the mean squared error (MSE) of the estimators of interest. Since
the MSE can be decomposed into a bias and a variance component, it
provides us with a natural criterion for combining the OLS and TSLS
estimators. Crucially, however, the proportion parameter weighting the relative
contributions of the two candidate estimators is adaptive, in the
sense that it depends on the properties of the data, and takes into
account the strength of the instruments. 
The idea of combining the OLS and TSLS estimators has been previously
discussed in the literature \citep{Angrist1995}. In particular, 
\citet{Sawa1973} has proposed an ``almost unbiased estimator'' for
simultaneous equations systems, which strikes a balance between two
different $k$-class estimators by weighting their relative
contributions using the sample size and the number of variables in the
model. Moreover, \citet{Angrist1995} have given an interpretation of
the limited information maximum likelihood (LIML) estimator as a
combination estimator, which relies on a weighting of the OLS and
TSLS estimators. Such combined estimators, however, do not attempt to
estimate the respective contributions of each estimator using the
data, as we have done in the paper at hand. The main contribution of
this article is therefore to provide a framework for estimating such
proportions in a data-informed adaptive manner. 

The paper is organized as follows. In section \ref{sec:cls}, we fix
the notation, and briefly recall the assumptions behind
OLS and TSLS estimation. We then show that these two estimators 
have complementary properties, in the sense that the OLS has minimal
variance, while the TSLS is asymptotically
unbiased. In section \ref{sec:cls proper}, we describe our proposed
convex estimator, and study its asymptotic properties, under the
assumption that the optimal proportion is known; whereas in section
\ref{sec:cls estimation}, we describe
a sample estimator of this proportion parameter. This framework
is then extended to other asymptotically unbiased estimators in a
third section. In section \ref{sec:sim}, these theoretical
results are tested using a range of different synthetic data sets.
The proposed methods are also applied to a classic data set from
econometrics in section \ref{sec:real}, and some conclusions are
provided in section \ref{sec:conclusion}. Finally, the proofs of all
the propositions in the paper are reported in the appendix. 

\section{Combining OLS and TSLS Estimators}\label{sec:cls}
\subsection{Ordinary Least Squares (OLS)}\label{sec:ols}
The model under scrutiny is described by the following linear relationship, 
\begin{equation}\label{eq:model}
    Y = X\bbe + \ep,
\end{equation}
where $X$ is a random row vector of order $1\times k$, and $\bbe$ is a
column vector of order $k\times 1$ representing the parameters of
interest, while $Y$ and $\ep$ are two real-valued random variables. Throughout, we will treat
both the error term, $\ep$, and the vector of predictors, $X$,
as random quantities, thereby allowing for possible correlations
between the $X_{j}$'s and $\ep$. For expediency, all random
variables, regardless of their dimensions, are simply denoted by
upper-case Roman letters. In general, a sample of
$n$ draws will be available from the model in equation
(\ref{eq:model}), such that 
\begin{equation}\notag
    y_{i} = \bx_{i}\bbe + \ep_{i}, \qquad\forall\;i=1,\ldots,n;
\end{equation}
where $\bx_{i}$ is again a row vector of order $1\times k$. This may also
be written using matrix notation as
\begin{equation}\notag
    \by = \bX\bbe + \bep,
\end{equation}
where $\by$ and $\bep$ are column vectors of order $n\times 1$, and
$\bX$ is a matrix of order $n\times k$. 

The estimation of the unknown vector of parameters, $\bbe$, can be performed by making some
standard assumptions about the moments of the different random
variables in (\ref{eq:model}), as commonly done in econometrics
\citep[see][]{Wooldridge2002}:
\begin{en}
  \item[(A1)] \textit{Exogeneity:} $\E[X\pri\ep]=\bzero$,
  \item[(A2)] \textit{Homoscedastitity:} $\E[\ep^{2}|X]=\sig^{2}$,
  \item[(A3)] \textit{Identification:} $\rank(\E[X\pri X])=k$;
\end{en}
with $\sig^{2}:=\E[\ep^{2}]$, and where $\E[X\pri X]$
represents a matrix of order $k\times k$. Under assumptions (A2) and
(A3), the OLS estimator behaves asymptotically as follows,
\begin{equation}\label{eq:ols}
     \wh\bbe_{n} := (\bX\pri\bX)^{-1}\bX\pri\by
     \stack{p}{\longrightarrow} 
     \E[X\pri X]^{-1}\E[X\pri Y] =: \wh\bbe.
\end{equation}
If assumption (A1) also holds, we say that the model in (\ref{eq:model}) is
\textit{exogenous}, and it then follows that the OLS 
estimator is asymptotically unbiased and consistent. That is, the
limit, $\wh\bbe$, can be shown to be equal to the true parameter, $\bbe$.
However, if assumption (A1) is violated, then the OLS estimator is
inconsistent. Thus, a model in which the vector of predictors has non-zero
correlations with the error term, $\ep$, is referred to as an
\textit{endogenous} model. 
\begin{figure}[t]
\centering\small
\begin{tikzpicture}
    \draw (-3,0) node[draw,inner sep=8pt](z){$Z$};
    \draw (0,0) node[draw,inner sep=8pt](x){$X$};
    \draw (+3,0) node[draw,inner sep=8pt](y){$Y$};
    \draw (+0,+1.5) node[draw,circle,inner sep=5pt](d){$\de$};
    \draw (+3,+1.5) node[draw,circle,inner sep=5pt](e){$\ep$};
    \draw (1.5,-1.75) node[draw,circle,inner sep=5pt](u){$U$};
    \draw[thick,->] (z) -- (x) node[midway,above]{$\bGa$};
    \draw[thick,->] (x) -- (y) node[midway,above]{$\bbe$};
    \draw[thick,->] (u) -- (x) node[midway,anchor=north east]{$\al$};
    \draw[thick,->] (u) -- (y) node[midway,anchor=north west]{$\al$};
    \draw[thick,->] (e) -- (y) node[midway,anchor=north west]{};
    \draw[thick,->] (d) -- (x) node[midway,anchor=north east]{};
\end{tikzpicture}
\caption{Graphical representation of the IV model described in
  equations (\ref{eq:model}) and (\ref{eq:model2}) in the presence of
  an unmeasured confounder $U$; where observed and latent variables
  are denoted by squares and circles, respectively. This graph
  corresponds to a two-level system of equations composed of
  $Y=X\bbe + U\al + \ep$, and $X=Z\bGa + U\al + \de$. When we assume
  that $\al\neq0$, condition (A1) is violated, and $X$ becomes
  endogenous.}
  \label{fig:model}
\end{figure}
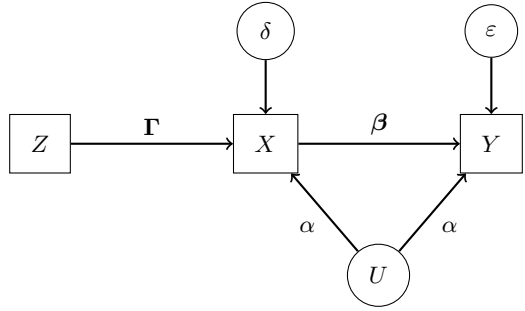

\subsection{Two-stage Least Squares (TSLS)}\label{sec:tsls}
The limitations of the OLS can be addressed by using a vector of
IVs, denoted $Z$. We will here
assume that $Z$ is a random row vector of order $1\times l$, with
$l\geq k$. This vector of instruments is used in a multivariate linear
equation of the form, 
\begin{equation}\label{eq:model2}
          X = Z\bGa + \de,
\end{equation}
where $\bGa$ is an unknown matrix of parameters of order $l\times k$,
and $X$ and $\de$ are random row vectors of order $1\times k$.
As before, we will usually work with a set of $n$ realizations from
this multivariate linear model expressed as follows,
\begin{equation}\label{eq:model2 realization}
    \bx_{i} = \bz_{i}\bGa + \bde_{i}, \qquad\forall\;i=1,\ldots,n;
\end{equation}
where $\bx_{i}$ and $\bde_{i}$ are $1\times k$ row vectors,
and $\bz_{i}$ is an $1\times l$ row vector. This can be concisely
expressed using matrix notation as
\begin{equation}\notag
    \bX = \bZ\bGa + \bD,
\end{equation}
where $\bZ$ and $\bD$ are matrices of order $n\times l$ and
$n\times k$, respectively. A graphical illustration of
the IV model is provided in figure \ref{fig:model}. 

When using the two-stage least squares (TSLS) estimator, we will make
the following additional assumptions about the random row vector of
instruments:
\begin{en}
  \item[(A4)] \textit{Exogeneity:} $\E[Z\pri\ep]=\bzero$,
  \item[(A5)] \textit{Homoscedastitity:} $\E[\ep^{2}|Z]=\sig^{2}$,
  \item[(A6)] \textit{Identification:} $\rank(\E[Z\pri Z])=l$,
    $\rank(\E[Z\pri X])=k$;
\end{en}
where, as before, $\sig^{2}:=\E[\ep^{2}]$. The TSLS estimator is then defined as 
\begin{equation}\notag
     \wti\bbe_{n} :=
     (\wh\bX{}\pri\wh\bX)^{-1}\wh\bX{}\pri\by,
\end{equation}
with $\wh\bX:=\bH_{z}\bX$ denoting the projection of the matrix of
predictors onto the column space of $\bZ$, and where
$\bH_{z}:=\bZ(\bZ\pri\bZ)^{-1}\bZ\pri$ is the hat
matrix of the multivariate regression in equation (\ref{eq:model2
  realization}). Under assumptions (A5) and (A6), the TSLS estimator
converges in probability to a non-stochastic vector, $\wti\bbe$,
defined as 
\begin{equation}\label{eq:tsls}
     \wti\bbe :=
     \big(\E[X\pri Z]\E[Z\pri Z]^{-1}\E[Z\pri X]\big)^{-1}
     \big(\E[X\pri Z]\E[Z\pri Z]^{-1}\E[Z Y]\big),
\end{equation}
such that $\wti\bbe_{n} \stack{p}{\longrightarrow}\wti\bbe$, as
described in \citet{Wooldridge2002}. Moreover, under assumption
(A4), this sequence of estimators can be shown to be asymptotically unbiased and
consistent with respect to the true vector of parameters, such that
$\wti\bbe=\bbe$. However, this gain in unbiasedness is compensated by
a larger variance of the TSLS estimator, as we discuss in the next
section. 

\subsection{Bias/Variance Trade-off}\label{sec:trade-off}
Under assumptions (A2-A6), the TSLS estimator is asymptotically
unbiased. By contrast, if assumption (A1) does not hold, then the OLS
estimator is asymptotically biased. However,
for finite $n$, the empirical variance of the TSLS estimator can be shown
to be larger than the one of the OLS estimator. We make these
observations formal by comparing the variance estimators
of the OLS and TSLS estimators. These are
\begin{equation}\label{eq:variance}
     \wh\var(\wh\bbe_{n}) :=
     \wh\sig^{2}_{n}(\bX\pri\bX)^{-1}, 
     \quad\te{and}\quad
     \wh\var(\wti\bbe_{n}) :=
     \wti\sig^{2}_{n}(\wh\bX{}\pri\wh\bX)^{-1};
\end{equation}
with the sample residual sums of squares (RSSs), $\wh\sig_{n}^{2}$ and $\wti\sig_{n}^{2}$,
being given by 
\begin{equation}\label{eq:rss}
     \wh\sig^{2}_{n} := 
     \frac{1}{n-k}\sum_{i=1}^{n}(y_{i}-\bx_{i}\wh\bbe_{n})^{2},
     \quad\te{and}\quad
     \wti\sig^{2}_{n} := 
     \frac{1}{n-k}\sum_{i=1}^{n}(y_{i}-\bx_{i}\wti\bbe_{n})^{2};
\end{equation}
for the OLS and TSLS estimators, respectively.

More remarkably, one can also approximate the bias of these two estimators. The
theoretical squared bias of a given arbitrary estimator, $\bbe^{\dag}_{n}$, is
defined as 
\begin{equation}\notag
     \bias^{2}(\bbe^{\dag}_{n}) := (\E[\bbe^{\dag}_{n}] -
     \bbe)(\E[\bbe^{\dag}_{n}] - \bbe)\pri,
\end{equation}
for every $n$. In the sequel, we will assume that the IVs under
scrutiny are valid instruments, such that assumption (A4) is true. 
Therefore, it follows that the TSLS estimator, $\wti\bbe_{n}$, is
known to be consistent, and can be used to construct a consistent
approximation of the bias of any arbitrary estimator, $\bbe_{n}^{\dag}$. For
large $n$, it follows that the squared bias of any such estimator can be
consistently estimated by 
\begin{equation}\label{eq:bias}
     \wh\bias{}^{2}(\bbe^{\dag}_{n}) := (\bbe^{\dag}_{n} -
     \wti\bbe_{n})(\bbe^{\dag}_{n} - \wti\bbe_{n})\pri.
\end{equation}
Observe that this empirical estimate of the bias gives a value of zero
for the TSLS estimator, for every $n$. This particular 
choice of empirical bias estimate can also be seen to be related to
the Hausman test, commonly used in econometrics for testing whether or
not the predictors of interest are exogenous
\citep{Hausman1978}. Indeed, the squared bias in equation
(\ref{eq:bias}) corresponds to the numerator of the Hausman test
statistic. 

Combining this empirical estimate of the bias with the standard
variance estimators in equation (\ref{eq:variance}), we can 
formalize our original observation about the trade-off between the
superiority of the TSLS estimator in terms of bias, and the
superiority of the OLS estimator in terms of variance. This result will
motivate our construction of a combined estimator, in which we will exploit the
respective strengths of the OLS and TSLS estimators, denoted by
$\wh\bbe_{n}$ and $\wti\bbe_{n}$, respectively. 
\begin{pro}\label{pro:bias-variance}
  Under assumptions (A2-A6), for every $n$, and for every realizations,
  $\by$, $\bX$, and $\bZ$, if both $\bX\pri\bX$ and $\wh\bX{}\pri\wh\bX$
  are invertible, then 
  \begin{en}
  \item[(i)] $\wh\bias{}^{2}(\wh\bbe_{n}) \succeq \wh\bias{}^{2}(\wti\bbe_{n})$, 
  \item[(ii)] $\wh\var(\wh\bbe_{n}) \preceq \wh\var(\wti\bbe_{n})$;
  \end{en}
  where $\succeq$ and $\preceq$ denote the positive semidefinite order
  for $k\times k$ matrices.
\end{pro}
Note that, in proposition \ref{pro:bias-variance}, we have requested
both $\bX\pri\bX$ and $\wh\bX{}\pri\wh\bX$ to be invertible. Indeed,
while assumptions (A3) and (A6) ensures that the stochastic limits
of these two matrices are invertible, this does not guarantee that
these matrices will be invertible for every $n$. 
Although these inequalities appear to be well-known, they
do not appear to have been formally proved in standard texts on
instrumental variables \citep[see, for
instance][]{Wooldridge2002,Davidson1993,Cameron2005}. A full proof of
this result is therefore provided in the appendix. 

Furthermore, the two statements in proposition
\ref{pro:bias-variance} can also be shown to hold in the stochastic
limit, as described in the following corollary. Note that this corollary is
trivially true for the variances of the OLS and TSLS estimators, since
both of these quantities converge to a zero matrix. A proof of this
result is provided in the appendix. 
\begin{cor}\label{cor:bias-variance}
  Under assumptions (A2-A6),
  \begin{en}
  \item[(i)] $\plim_{n}\wh\bias{}^{2}(\wh\bbe_{n}) \succeq \plim_{n}\wh\bias{}^{2}(\wti\bbe_{n})$, 
  \item[(ii)] $\plim_{n}\wh\var(\wh\bbe_{n}) \preceq \plim_{n}\wh\var(\wti\bbe_{n})$;
  \end{en}
  where, as before, $\wh\bbe_{n}$ and $\wti\bbe_{n}$ denote the OLS
  and TSLS estimators, respectively. 
\end{cor}
The inequalities in proposition \ref{pro:bias-variance} indicate that
it may be fruitful to compare the MSEs of these two estimators for finite
$n$. Clearly, since the bias tends to dominate the MSE asymptotically, it follows that the
TSLS should exhibit a smaller level of bias as $n$ goes to infinity. Nonetheless,
for finite samples, the OLS may yield a smaller MSE than its two-stage
counterpart, due to its greater efficiency. Therefore, one may try to
strike a balance between the relative strengths of these two types of
estimators, using the sample MSE as a criterion. 

\subsection{Convex Least Squares (CLS)}\label{sec:cls proper}
In this section and in the rest of this paper, we now assume that (A2-A6)
hold. In addition, we also assume that the random vectors,
$\wh\bbe_{n}$ and $\wti\bbe_{n}$, are well-behaved, in the sense that
they are elementwise squared-integrable for every $n$.
Under these assumptions, we propose an estimator, denoted $\bar\bbe_{n}(\pi)$, which is
defined as a convex combination of the OLS and TSLS estimators, such that
\begin{equation}\label{eq:cls}
    \bar\bbe_{n}(\pi) := \pi\wh\bbe_{n} + (1-\pi)\wti\bbe_{n},
\end{equation}
for every $\pi\in[0,1]$. The \textit{proportion parameter}, $\pi$, controls the
respective contributions of the OLS and TSLS estimators. This 
parameter is selected in order to minimize the trace of the
theoretical MSE of the corresponding CLS estimator,
\begin{equation}\notag
    \mse(\bar\bbe_{n}(\pi)) 
         = \E\!\big[\!\big.(\bar\bbe_{n}(\pi)-\bbe)
                   (\bar\bbe_{n}(\pi)-\bbe)\pri\big], 
\end{equation}
where $\bbe\in\R^{k}$ is the true parameter of interest and the MSE is
a $k\times k$ matrix. 

The MSE automatically strikes a trade-off between the unbiasedness of
the TSLS estimator and the efficiency of the OLS estimator. Indeed,
this criterion can be decomposed into a variance and a bias component,
such that 
\begin{equation}\notag
    \mse(\bar\bbe_{n}(\pi)) 
    = \var(\bar\bbe_{n}(\pi)) + \bias^{2}(\bar\bbe_{n}(\pi)).
\end{equation}
Therefore, in the light of proposition \ref{pro:bias-variance}, this
criterion constitutes a natural choice for combining these two types of
estimators.

The MSE of the CLS estimator, $\mse(\pi\wh\bbe_{n} +
(1-\pi)\wti\bbe_{n})$, can be expressed as the weighted sum of the
MSEs of the OLS and TSLS estimators, as well as a
\textit{cross-squared-error} (CSE) term between these two estimators, 
\begin{equation}\label{eq:mse pi}
    \pi^{2}\mse(\wh\bbe_{n}) 
    + 2\pi(1-\pi)\cse(\wh\bbe_{n},\wti\bbe_{n})
    + (1-\pi)^{2}\mse(\wti\bbe_{n}),
\end{equation}
where the cross-term is defined as follows, 
\begin{equation}\notag
    \cse(\wh\bbe_{n},\wti\bbe_{n})
    := \E[(\wh\bbe_{n} - \bbe)(\wti\bbe_{n} - \bbe)\pri].
\end{equation}
By analogy with the MSE, we can also decompose the CSE 
into a covariance term and a squared \textit{cross-bias} term, denoted
$\bias^{2}(\wh\bbe_{n},\wti\bbe_{n})$, such that 
\begin{equation}\notag
     \cse(\wh\bbe_{n},\wti\bbe_{n})
     =  \cov(\wh\bbe_{n},\wti\bbe_{n})
     + \bias^{2}(\wh\bbe_{n},\wti\bbe_{n}),
\end{equation}
where the squared cross-bias term is 
$\bias^{2}(\wh\bbe_{n},\wti\bbe_{n}):=(\E[\wh\bbe_{n}] -
\bbe)(\E[\wti\bbe_{n}] - \bbe)\pri$.

The true (or theoretical) proportion parameter, $\pi$, is defined as the value that minimizes
the trace of the theoretical MSE of the CLS estimator. Note that we
are here considering a sequence of parameters, $\pi_{n}$, since this
definition may yield a different proportion for different sample
sizes. Therefore, for every $n$, the target proportion parameter
is given by
\begin{equation}\label{eq:pi}
    \pi_{n} := \argmin_{\pi\in[0,1]} \tr\mse(\bar\bbe_{n}(\pi)).
\end{equation}
Crucially, this parameter is available in closed-form, and it can also be
shown to be unique, since the trace of the
theoretical MSE of $\bar\bbe_{n}$ is a convex
function of $\pi$. This statement is made formal in the following
proposition, which is proved using the aforementioned
decomposition of the MSE of the CLS estimator. The proportion
parameter is only non-unique when the square-root of the trace of the
MSEs the OLS and TSLS estimators are identical. This quantity, denoted by
$(\tr\mse(\bbe^{\dag}_{n}))^{1/2}$ for every estimator
$\bbe_{n}^{\dag}$, will be referred to as the RMSE of
$\bbe_{n}^{\dag}$, in the sequel. See appendix \ref{sec:proof} for a
proof of this minimization. 
\begin{pro}\label{pro:pi}
   For every $n$, the proportion parameter 
   defined in equation (\ref{eq:pi}) is given by
   \begin{equation}\notag
        \pi_{n} = 
         \frac{\tr(\mse(\wti\bbe_{n}) -
         \cse(\wh\bbe_{n},\wti\bbe_{n}))}{\tr(\mse(\wti\bbe_{n}) -
         2\cse(\wh\bbe_{n},\wti\bbe_{n}) + \mse(\wh\bbe_{n}))}. 
   \end{equation}
   It is unique whenever the RMSEs of the OLS and TSLS estimators are
   not equal. 
\end{pro}
Finally, we can verify that the CLS estimator based on the true
proportion $\pi_{n}$ has an MSE, which is lower or equal to the
MSEs of the OLS and TSLS estimators. Note that this inequality is not
immediate from our definition of $\pi_{n}$, as we need to control for
the additional CSE term in equation \ref{eq:mse pi}. A proof of this
proposition is also provided in the appendix. 
\begin{pro}\label{pro:mse optimal}
  The CLS estimator based on the true proportion, $\pi_{n}$, satisfies
  \begin{equation}\notag
      \tr\mse(\bar\bbe_{n}(\pi_{n}))\leq
            \tr\min\lb\mse(\wh\bbe_{n}),\mse(\wti\bbe_{n})\rb,
  \end{equation}
  for every $n$.
\end{pro}
Observe that this result holds in greater generality, since
the OLS and TSLS estimators could be replaced by other candidate
estimators. In the next section, we describe how to estimate the
proportion parameter in an adaptive manner for this particular choice
of estimators; whereas in section \ref{sec:other}, we consider how the
CLS can accommodate other estimators. 

\subsection{CLS Estimation}\label{sec:cls estimation}
When evaluating $\pi_{n}$ from a particular data set, we estimate this
parameter by minimizing the trace of an empirical estimate of the
theoretical MSE of the CLS estimator. A consistent estimator of the MSE
can be obtained by setting the true parameter, $\bbe$, to be equal to the TSLS
estimator, $\wti\bbe_{n}$. Thus, for every $\pi\in[0,1]$, our proposed
empirical MSE is given by
\begin{equation}\label{eq:empirical mse}
    \wh\mse(\bar\bbe_{n}(\pi)) 
    = \wh\var(\bar\bbe_{n}(\pi)) + \wh\bias{}^{2}(\bar\bbe_{n}(\pi)),
\end{equation}
where $\wh\bias(\bar\bbe_{n}(\pi)):=\bar\bbe_{n}(\pi) - \wti\bbe_{n}$. 
That is, we here use the TSLS estimator as a consistent estimator of the true
parameter, $\bbe$. To approximate the population variance of the CLS
estimator, we can use a combination of the empirical estimates of the
variances of the two estimators of interest, such that 
\begin{equation}\label{eq:cls variance}
    \wh\var(\bar\bbe_{n}(\pi)) 
      = \pi^{2}\wh\var(\wh\bbe_{n}) 
      + 2\pi(1-\pi)\wh\cov(\wh\bbe_{n},\wti\bbe_{n})
      + (1-\pi)^{2}\wh\var(\wti\bbe_{n}),
\end{equation}
where the empirical variances of the OLS and TSLS estimators have
already be given in equation (\ref{eq:variance}); and where the
covariance term takes the following form, 
\begin{equation}\notag
     \wh\cov(\wh\bbe_{n},\wti\bbe_{n}) 
     := \bar\sig^{2}_{n}(\bX\pri\bX)^{-1}(\bX\pri\wh\bX)(\wh\bX{}\pri\wh\bX)^{-1}
      = \bar\sig^{2}_{n}(\bX\pri\bX)^{-1};
\end{equation}
with as before, $\wh\bX:=\bH_{z}\bX$, and in which the second equality is
obtained by using the idempotency of $\bH_{z}$. Moreover, the cross-RSS, denoted
$\bar\sig^{2}_{n}$, is given by 
\begin{equation}\notag
     \bar\sig^{2}_{n} := \frac{1}{n-k}\sum_{i=1}^{n}
        (y_{i} - \bx_{i}\wh\bbe_{n})(y_{i} - \bx_{i}\wti\bbe_{n}),
\end{equation}
which can be compared to the RSSs of the OLS and TSLS estimators in
equation (\ref{eq:rss}). 

The second term in equation (\ref{eq:empirical mse}) consists of the
empirical bias of the CLS estimator. As for the empirical variance,
the bias can be estimated by using the TSLS estimator to replace the
unknown true parameter, such that 
\begin{equation}\notag
    \wh\bias{}^{2}(\bar\bbe_{n}(\pi)) = 
    \pi^{2}\wh\bias{}^{2}(\wh\bbe_{n}) 
      + 2\pi(1-\pi)\wh\bias{}^{2}(\wh\bbe_{n},\wti\bbe_{n})
      + (1-\pi)^{2}\wh\bias{}^{2}(\wti\bbe_{n}),
\end{equation}
where the empirical biases of the OLS and TSLS estimators are
estimated as in equation (\ref{eq:empirical mse}). Since we have set
the bias of $\wti\bbe_{n}$ to zero, it follows that the cross-bias
term also eliminates, Thus, the empirical bias of the CLS estimator
becomes proportional to the one of the OLS estimator, such that we obtain
\begin{equation}\label{eq:cls bias}
    \wh\bias{}^{2}(\bar\bbe_{n}(\pi)) = 
    \pi^{2}\wh\bias{}^{2}(\wh\bbe_{n}).
\end{equation}
The empirical estimate of the MSE in equation (\ref{eq:empirical mse})
can be shown to be consistent, as described in the
following proposition, which is proved in the appendix. Observe that
this statement holds for every arbitrary proportion comprised between
$0$ and $1$.
\begin{pro}\label{pro:mse}
   For every $\pi$, $\bar\bbe_{n}(\pi)\stack{p}{\longrightarrow}
   \bar\bbe(\pi):=\pi\wh\bbe + (1-\pi)\wti\bbe$, where $\wh\bbe$ and
   $\wti\bbe$ are defined as in equations (\ref{eq:ols}) and
   (\ref{eq:tsls}), respectively. Moreover, 
   \begin{equation}\notag
      \wh\mse(\bar\bbe_{n}(\pi))\stack{p}{\longrightarrow}
      \mse(\bar\bbe(\pi)).
   \end{equation}   
\end{pro}

As for the true proportion parameter, $\pi_{n}$, which minimizes the trace
of the theoretical MSE, the proportion estimator, $\wh\pi_{n}$, which
minimizes the trace of the empirical MSE; is also 
available in closed-form. We thus obtain the following result, as a
corollary to proposition \ref{pro:pi}. Observe that, since the bias of
the TSLS estimator is zero under our estimation framework, it follows
that the MSE of the TSLS reduces to the variance of that estimator,
and that the CSE term reduces to the covariance of the two estimators
of interest. 
\begin{cor}\label{cor:pi}
   The estimator of the proportion parameter, $\pi\in[0,1]$, defined as
   $\wh\pi_{n}:=\argmin\tr\wh\mse(\bar\bbe_{n}(\pi))$, in which
   the $\wh\mse$ is defined as in equation (\ref{eq:empirical mse});
   satisfies, 
   \begin{equation}\notag
        \wh\pi_{n} = \frac{\tr(\wh\var(\wti\bbe_{n}) -
         \wh\cov(\wh\bbe_{n},\wti\bbe_{n}))}{\tr(\wh\var(\wti\bbe_{n}) -
         2\wh\cov(\wh\bbe_{n},\wti\bbe_{n}) + \wh\mse(\wh\bbe_{n}))}.     
   \end{equation}
\end{cor}
We have here emphasized the estimation of the proportion parameter. 
Our original motivation, however, for constructing the CLS estimator
centered on producing an estimator, which would minimize the MSE. This
can be achieved by estimating the CLS estimator, $\bar\bbe_{n}(\pi)$, at the value
of the estimated proportion, $\wh\pi_{n}$, thereby producing $\bar\bbe_{n}(\wh\pi_{n})$.
We thus conclude this section by verifying that this particular CLS estimator
behaves as expected asymptotically, in the sense that it is both
weakly and MSE consistent. 
\begin{pro}\label{pro:cls consistency}
  Under assumptions (A2-A6), the CLS estimator, $\bar\bbe_{n}(\wh\pi_{n})$, satisfies 
  (i) $\bar\bbe_{n}(\wh\pi_{n})\stack{p}{\longrightarrow}\bbe$, and (ii)
  $\bar\bbe_{n}(\wh\pi_{n})\stack{L^{2}}{\longrightarrow}\bbe$.
\end{pro}
The proof of this proposition follows from the inequalities reported
in proposition \ref{pro:mse optimal}, combined with the fact that the
TSLS estimator is both weakly and MSE consistent. See appendix
\ref{sec:proof} for details. 
Observe that the CLS framework relies on the existence of the first
two moments of the TSLS estimator. For finite $n$, \citet{Kinal1980}
has shown that the TSLS estimator only possesses first and second moments when
$l\geq k+2$. Asymptotically, however, such moments always exist. As
for the TSLS therefore, we are thus considering an estimator, which is
solely asymptotically well-identified. This particular issue is
further discussed in section \ref{sec:conclusion}.

\subsection{Bootstrap CLS Variance}\label{sec:bootstrap variance}
We now turn to the question of estimating the variance of our proposed
CLS estimator. In equation (\ref{eq:cls variance}), we have described the
variance of $\bar\bbe_{n}(\pi)$, for every $\pi$. This quantity was
then used in our proposed empirical MSE, in order to obtain a sample
estimate of $\pi$. However, the variance formula in equation
(\ref{eq:cls variance}) does not take into account the variability
associated with the choice of $\pi_{n}$. The derivation of a
closed-form estimator for the variance of $\bar\bbe_{n}(\wh\pi_{n})$
is beyond the scope of this paper. However, in practice, the variance of the
CLS estimator can be computed using the bootstrap by sampling with
replacement from the triple $(\by,\bX,\bZ)$, and producing $B$ bootstrap
samples denoted $(\by\as_{b},\bX\as_{b},\bZ\as_{b})$, with
$b=1,\ldots,B$. We are here adopting the framework described by previous
researchers, who have also used the bootstrap in the context of
IV estimation \citep[see for example][]{Wong1996}.  

Specifically, each bootstrap sample is constructed by sampling $n$ cases with
replacement from the collection of triples $(y_{i},\bx_{i},\bz_{i})$, with
$i=1,\ldots,n$. These bootstrap samples are then used to produce the
bootstrap distribution of the CLS estimator.
That is, for each of these bootstrap samples, we compute the
CLS estimator, $\bar\bbe_{nb}:=\bar\bbe_{nb}(\wh\pi_{nb})$, which
leads to the following bootstrap variance estimator, 
\begin{equation}\notag
     \wh\var{}\as(\bar\bbe_{n}) := 
        \frac{1}{B-1}\sum_{b=1}^{B}(\bar\bbe_{nb}\as-\E\as[\bar\bbe_{n}])
                                 (\bar\bbe_{nb}\as-\E\as[\bar\bbe_{n}])\pri,
\end{equation}
where $\E\as[\bar\bbe_{n}]:=\sum_{b}\bar\bbe_{nb}\as/B$ denotes the
bootstrap mean. In our real-world data set application in section
\ref{sec:real}, we will report the variance of the CLS and its
confidence interval using the bootstrap. 

One of the limitations of our discussion thus far is the presence of a
finite-sample bias in the TSLS estimator. In the next section, we consider
other consistent estimators, which could be articulated within our
framework by being substituted to the TSLS estimator. Indeed, every
asymptotically unbiased estimator could be used to replace the TSLS
estimator in the previous results. 

\section{Extensions to Other Unbiased Estimators}\label{sec:other}
Our proposed convex combination of least squares estimators essentially
relies on the choice of an asymptotically unbiased estimator. Under
standard assumptions on the properties of the IVs under scrutiny,
the TSLS estimator satisfies this criterion. This choice was
mainly motivated by computational
considerations. The empirical variance for the TSLS estimator is
indeed well-known and can easily be manipulated. We now extend 
the CLS framework, in order to accommodate other asymptotically
unbiased estimators. The corresponding MSE can be empirically
estimated using the bootstrap at a greater computational cost, but
without additional theoretical complications. The resulting estimator
will thus be referred to as the bootstrap CLS. 

\subsection{Jackknife IV Estimator}\label{sec:jive}
An ideal replacement for the TSLS estimator is the jackknife IV
estimator (JIVE), which we now describe. This estimator was originally
introduced by \citet{Angrist1995} in order to reduce the finite-sample
bias of the TSLS estimator, when applied to a large number of
instruments. Indeed, the TSLS estimator tends to behave poorly as the
number of instruments increases. We briefly outline this method in the present
section. See \citet{Angrist1999} for an exhaustive description. 
Let the estimator of the regression parameter in the
first-level equation in model (\ref{eq:model2}) be denoted by 
\begin{equation}\notag
     \wh\bGa := (\bZ\pri\bZ)^{-1}(\bZ\pri\bX),
\end{equation}
which is of order $l\times k$. The matrix of predictors, $\bX$, projected
onto the column space of the instruments is then given by
$\wh\bX=\bZ\wh\bGa$. The jackknife IV estimator (JIVE) proceeds by
estimating each row of $\wh\bX$ without using the corresponding data
point. That is, the $i\tth$ row in the jackknife matrix, $\wh\bX_{J}$,
is estimated without using the $i\tth$ row of $\bX$. 

This is conducted as follows. For every $i=1,\ldots,n$, we first compute
\begin{equation}\notag
     \wh\bGa_{(i)} := (\bZ_{(i)}\pri\bZ_{(i)})^{-1}(\bZ_{(i)}\pri\bX_{(i)}),
\end{equation}
where $\bZ_{(i)}$ and $\bX_{(i)}$ denote matrices $\bZ$ and $\bX$ after
removal of the $i\tth$ row, such that these two matrices are of order
$(n-1)\times l$ and $(n-1)\times k$, respectively. Then, the
matrix $\wh\bX_{J}$ is constructed by stacking these jackknife
estimates of $\wh\bGa$, after they have been pre-multiplied by the
corresponding rows of $\bZ$, 
\begin{equation}\notag
     \wh\bX_{J} :=
     \begin{vmatrix}
          \bz_{1}\wh\bGa_{(1)} \\ \vdots \\ \bz_{n}\wh\bGa_{(n)}
     \end{vmatrix},
\end{equation}
where each $\bz_{i}$ is an $l$-dimensional row vector. The JIVE
estimator is then obtained by replacing $\wh\bX$ with $\wh\bX_{J}$ in
the standard formula of the TSLS, such that 
\begin{equation}\notag
      \wti\bbe_{J} := (\wh\bX_{J}{}\pri\bX)^{-1}(\wh\bX_{J}{}\pri\by). 
\end{equation}
In this paper, we have additionally made use of the computational
formula suggested by \citet{Angrist1999}, in which each row of
$\wh\bX_{J}$ is calculated using 
\begin{equation}\notag
     \bz_{i}\wh\bGa_{(i)} = \frac{\bz_{i}\wh\bGa -
       h_{i}\bx_{i}}{1-h_{i}},
\end{equation}
where $\bz_{i}\wh\bGa_{(i)}$, $\bz_{i}\wh\bGa$ and $\bx_{i}$ are
$k$-dimensional row vectors; and with $h_{i}$ denoting
the leverage of the corresponding data point in the first-level
equation of our model, such that each $h_{i}$ is defined as
$\bz_{i}(\bZ\pri\bZ)^{-1}\bz_{i}\pri$.

\subsection{Bootstrap CLS Estimation}\label{sec:bootstrap}
When replacing the TSLS estimator with an arbitrary estimator, such as
the JIVE, some of the quantities required for estimating the
proportion, $\pi_{n}$, need not be available in closed-form. 
However, such quantities can be straightforwardly estimated using the
bootstrap, as was done for the variance of the CLS estimator in
section \ref{sec:bootstrap variance}.

We can indeed approximate the unknown joint distribution, $F(Y,X,Z)$, with its
bootstrap estimate, $F\as$, using the straightforward sampling scheme
described in section \ref{sec:bootstrap variance}. As before, we thus generate
$B$ bootstrap samples, denoted $(\by\as_{b},\bX\as_{b},\bZ\as_{b})$,
from $F\as$. These bootstrap samples are then used to produce the
bootstrap distributions of the OLS estimator and the unbiased
estimator of interest such as the JIVE; and the 
first and second moments of these estimators are computed.
Thus, for every unbiased estimator, $\bbe^{\dag}_{n}$, and given
the OLS estimator, $\wh\bbe_{n}$, we construct a bootstrap estimate of
the MSE of the corresponding CLS estimator, such that for every $\pi$,
we define 
\begin{equation}\notag
    \wh\mse{}\as(\bar\bbe_{n}(\pi))
    := \E{}\as\Big[(\bar\bbe_{n}(\pi) - \E{}\as[\bbe^{\dag}_{n}])
                  (\bar\bbe_{n}(\pi) - \E{}\as[\bbe^{\dag}_{n}])\pri\Big].
\end{equation}
As in section \ref{sec:bootstrap variance}, the operator, $\E{}\as$,
denotes the expectation over the bootstrap estimate of $F$. Similarly
to the MSE decomposition in equation (\ref{eq:mse pi}), the bootstrap
estimate of the MSE can be decomposed into the following components, 
\begin{equation}\notag
    \pi^{2}\wh\mse{}\as(\wh\bbe_{n}) 
    + 2\pi(1-\pi)\wh\cse{}\as(\wh\bbe_{n},\bbe^{\dag}_{n})
    + (1-\pi)^{2}\wh\var{}\as(\bbe^{\dag}_{n}),
\end{equation}
where the bootstrap estimate of the MSE of $\bbe^{\dag}_{n}$ was reduced to
$\wh\var{}\as(\bbe^{\dag}_{n})$, since the estimator, $\bbe^{\dag}_{n}$, is assumed to
be unbiased for every $n$. Moreover, as in equation (\ref{eq:cls bias}),
the bootstrap estimate of the bias of the CLS estimator is
proportional to the bootstrap bias of the OLS estimator, such that we
have 
\begin{equation}\notag
     \wh\bias{}\as(\bar\bbe_{n}(\pi))
      = \pi^{2}\wh\bias{}\as(\wh\bbe_{n}).
\end{equation}

The boostrap estimate of the proportion, denoted $\wh\pi\as_{n}$, is
then given by a formula analogous to the one described in corollary \ref{cor:pi},
in which each empirical moment is replaced by its bootstrap
equivalent. This allows us to show that, for every choice of
asymptotically unbiased estimator, $\bbe^{\dag}_{n}$, the resulting
bootstrap CLS estimator, $\bar\bbe_{n}(\wh\pi\as_{n})$, achieves minimal bootstrap
MSE amongst its constituent estimators. A proof of this corollary is
provided in the appendix. It relies on the same arguments employed in the
proof of the optimality of the CLS estimator in proposition \ref{pro:mse optimal}.
\begin{cor}\label{cor:sample optimal}
  For every asymptotically unbiased estimator $\bbe^{\dag}_{n}$, the
  bootstrap CLS estimator of $\wh\bbe_{n}$ and $\bbe^{\dag}_{n}$, based on the
  bootstrap proportion, $\wh\pi_{n}\as$, satisfies
  \begin{equation}\notag
      \tr\wh\mse{}\as(\bar\bbe_{n}(\wh\pi_{n}\as))\leq
      \tr\min\lb\wh\mse{}\as(\wh\bbe_{n}),\wh\mse{}\as(\bbe_{n}^{\dag})\rb,
  \end{equation}
  for every $n$.      
\end{cor}

\section{Data Simulations}\label{sec:sim}
We here produce synthetic data sets with different number of
instruments. All of the models considered in this section are based on
a univariate endogenous variable, $X$, without any additional
covariate in the second-level equation. In Model I, we describe a
simple Gaussian model with a single valid instrument; whereas in Model
II, we consider a similar statistical model comprising $l=10$
uncorrelated instruments. 

\subsection{Simulation Model I}\label{sec:sim model i}
    Synthetic data sets were created from the following
    two-level model. We are here focusing on a univariate model composed of 
    a single predictor, $X$, and a single instrument, $Z$. For every
    $i=1,\ldots,n$, the two levels of the model are
    \begin{equation}\label{eq:model i}
      \begin{aligned}
        y_{i} &= x_{i}\be + u_{i}\al + \ep_{i}, \\
        x_{i} &= z_{i}\ga + u_{i}\al + \de_{i};
      \end{aligned}
    \end{equation}
    where $\al$ controls the degree of endogeneity of $X$, and $\ga$
    controls the amount of covariance between $X$ and the instrument
    $Z$, such that $\ga$ can be interpreted as the strength of the
    instrument. We wish to keep the marginal variances of the
    $Y_{i}$'s and $X_{i}$'s constant, while varying the values of
    $\al$ and $\ga$. This is achieved by defining the variances of the
    error terms, $\ep_{i}$ and $\de_{i}$, as functions of $\al$ and $\ga$.
    In doing so, we simplify the interpretation of
    $\be$, which becomes a standardized regression coefficient, whenever $\ga=0$. Throughout
    these simulations, the true parameter of interest will be set to
    be $\be=1/2$. A graphical representation of this model has been given in
    figure \ref{fig:model}.

    The model is thus standardized by setting the marginal variances
    of the $Y_{i}$'s and $X_{i}$'s to one, such that $\var(Y_{i})=\var(X_{i})=1$;
    and by generating the $Z_{i}$'s and $U_{i}$'s from a standard
    normal distribution, such that 
    \begin{equation}\notag
        Z_{i},U_{i} \stack{\iid}{\sim} N(0,1),\qquad \for
        i=1,\ldots,n. 
    \end{equation}
    The marginal variances of these random variables will be
    denoted by $\sig_{z}^{2}:=\var[Z_{i}]$ and
    $\sig^{2}_{u}:=\var[U_{i}]$, respectively. 
    The two remaining variances can then be defined as functions of
    the different regression parameters. For the second-level equation,
    we have 
    \begin{equation}\label{eq:sigep}
        \ep_{i} \stack{\iid}{\sim} N(0,\sig^{2}_{\ep}(\al)),\qquad 
        \sig^{2}_{\ep}(\al) := 3/4 - 2\al^{2},
    \end{equation}
    which follows from the constraint $\var[Y_{i}]=1$, and from the
    decomposition,
    \begin{equation}\notag
        \var(Y_{i}) = \be^{2}\var(X_{i}) + \al^{2}\var(U_{i}) +
        2\be\al\cov(X_{i},U_{i}) + \var(\ep_{i}).
    \end{equation}
    Using the linear independence of $Z$ and $U$, the covariance term becomes
    $\cov(X_{i},U_{i})=\al\sig^{2}_{u}$. Moreover, from our choice of
    variances for $X_{i}$ and $U_{i}$, we also obtain
    $\var(Y_{i})=\be^{2} + \al^{2} + 2\be\al^{2}+\sig^{2}_{\ep}$. 
    Fixing the variance of $Y_{i}$ to unity and using our choice of $\be$,
    this yields the definition of $\sig^{2}_{\ep}(\al)$ given in
    equation (\ref{eq:sigep}). Moreover, observe that the positiveness of
    $\sig^{2}_{\ep}$ produces an upper bound for $\al$, which is
    given by $\al<\sqrt{3/8}$. 

    Similarly, we can ensure that the marginal variances of the
    $X_{i}$'s are also constant, irrespective of the choice of $\al$
    and $\ga$, by controlling the variances of the $\de_{i}$'s. Thus,
    we set
    \begin{equation}\notag
        \de_{i} \stack{\iid}{\sim} N(0,\sig^{2}_{\de}(\al,\ga)),\qquad 
        \sig^{2}_{\de}(\al,\ga) := 1 - (\ga^{2}+\al^{2}). 
    \end{equation}
    This specification ensures that the variances
    of the $X_{i}$'s are constant with $\sig^{2}_{x}=1$. 
    That is, since by assumption $\cov(Z_{i},U_{i})=0$, and
    using the fact that for uncorrelated variables, the Bienaym\'e
    formula states that $\var(\sum_{j}X_{j}) = \sum_{j}\var(X_{j})$;
    it then follows that for every $i=1,\ldots,n$, we obtain
    \begin{equation}\notag
        \var(X_{i}) = \ga^{2}\var(Z_{i}) +
        \al^{2}\var(U_{i}) + \var(\de_{i}),
    \end{equation}
    which gives, $\var(\de_{i}) = 1 - \ga^{2}\sig^{2}_{z} -
    \al^{2}\sig^{2}_{u}$, as required. Moreover, note that we must
    have $\ga<\sqrt{1-\al^{2}}$ in order to ensure that
    $\sig^{2}_{\de}>0$. Using our previous bound for $\al$, which
    states that $\al<\sqrt{3/8}$, it then follows that $\ga<\sqrt{5/8}$. 
    
    Altogether, we have therefore fixed the variances of the
    $Y_{i}$'s, $X_{i}$'s, $U_{i}$'s, and $Z_{i}$'s to unity; and by
    assumption, the instrument is deemed valid in the sense that
    $\cov(Z_{i},U_{i})=0$. From these
    standardizations, it follows that for every
    $\al\in[0,\sqrt{3/8})$, and for every
    $\ga\in(0,\sqrt{1-\al^{2}})$, the correlations of the $X_{i}$'s
    with the $U_{i}$'s and $Z_{i}$'s are controlled by the two
    simulation parameters, $\al$ and $\ga$:
    \begin{equation}\notag
        \Cor(X_{i},U_{i}) = \al, 
        \qquad\te{and}\qquad 
        \Cor(X_{i},Z_{i}) = \ga,
    \end{equation}
    which respectively represent the \textit{magnitude of the confounding} and the
    \textit{strength of the instrument}. In addition, the correlations of the
    $Y_{i}$'s with the $U_{i}$'s and the $Z_{i}$'s are also controlled
    by a combination of these parameters. These correlations are
    respectively given by $\Cor(Y_{i},U_{i})=\be\al+\al$, and
    $\Cor(Y_{i},Z_{i})=\be\ga$. Finally, the correlation between the
    outcome and the endogenous variable satisfies 
    \begin{equation}\notag
        \Cor(Y_{i},X_{i}) = \be + \al^{2}.
    \end{equation}
    Therefore, in the absence of any confounding effect, $\be$ can be
    interpreted as the correlation coefficient between the $Y_{i}$'s
    and the $X_{i}$'s. 
\begin{figure}[t]
   \centering
   \includegraphics[width=13cm]{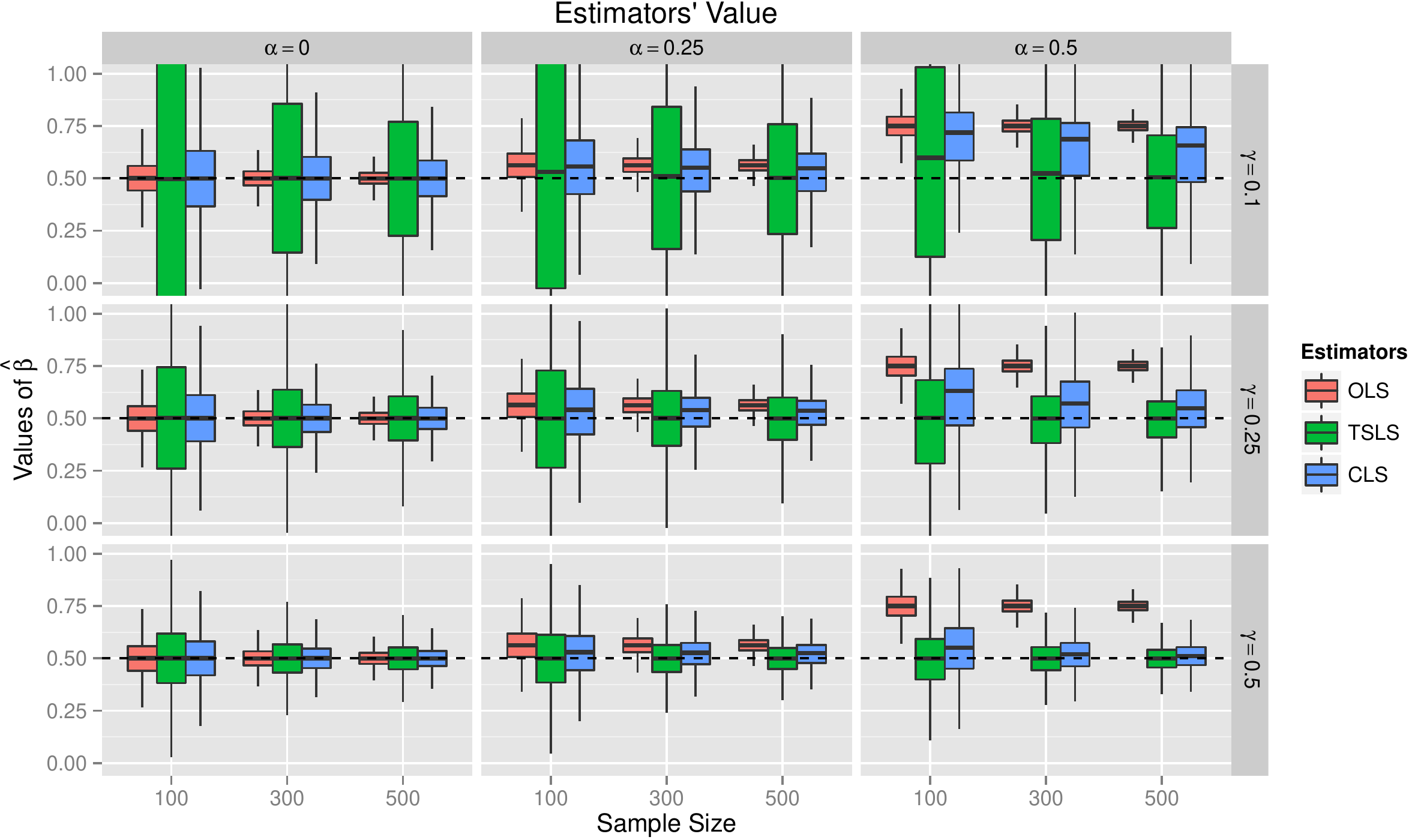}\\
   \caption{Monte Carlo distributions of the estimators' values under
     three different levels of confounding, $\al=\Cor(X_{i},U_{i})$;
     and for three different levels of instrument's strength,
     $\ga=\Cor(X_{i},Z_{i})$. In each
     panel, the sample size varies between $n=100$ and $n=500$.
     We here compare the OLS, TSLS and CLS estimators with respect to
     the true parameter $\be=1/2$, whose value is indicated by a
     dashed line. These simulations are based on $10^{5}$
     iterations for each scenario. The boxplots are here
     centered at the median, and the upper and lower hinges correspond
     to the first and third quartiles. 
     \label{fig:beta}}
\end{figure}
\begin{figure}[t]
  \centering
  \includegraphics[width=13cm]{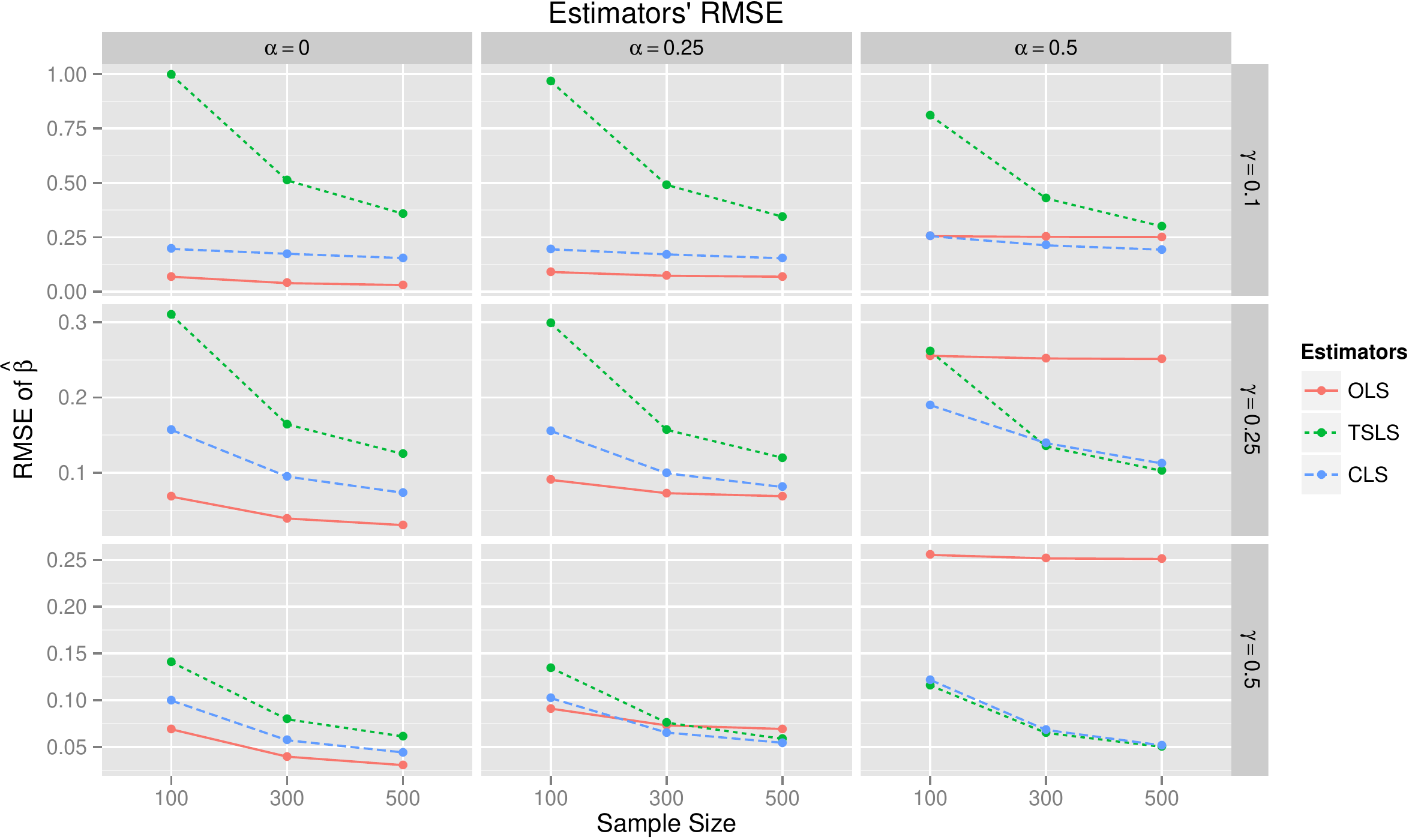}
   \caption{Monte Carlo estimates of the root mean squared errors (RMSEs)
     of the three estimators of interest under the simulation scenarios described
     in figure \ref{fig:beta}. As predicted, the RMSE of the proposed CLS method strikes a
     trade-off between its two constituent estimators. Indeed, under
     small $\al$, the CLS's RMSE tends towards the RMSE of the OLS estimator;
     whereas under large $\ga$, it tends towards the RMSE of the TSLS estimator.
     \label{fig:mse}}
\end{figure}
\begin{figure}[t]
  \centering
  \includegraphics[width=13cm]{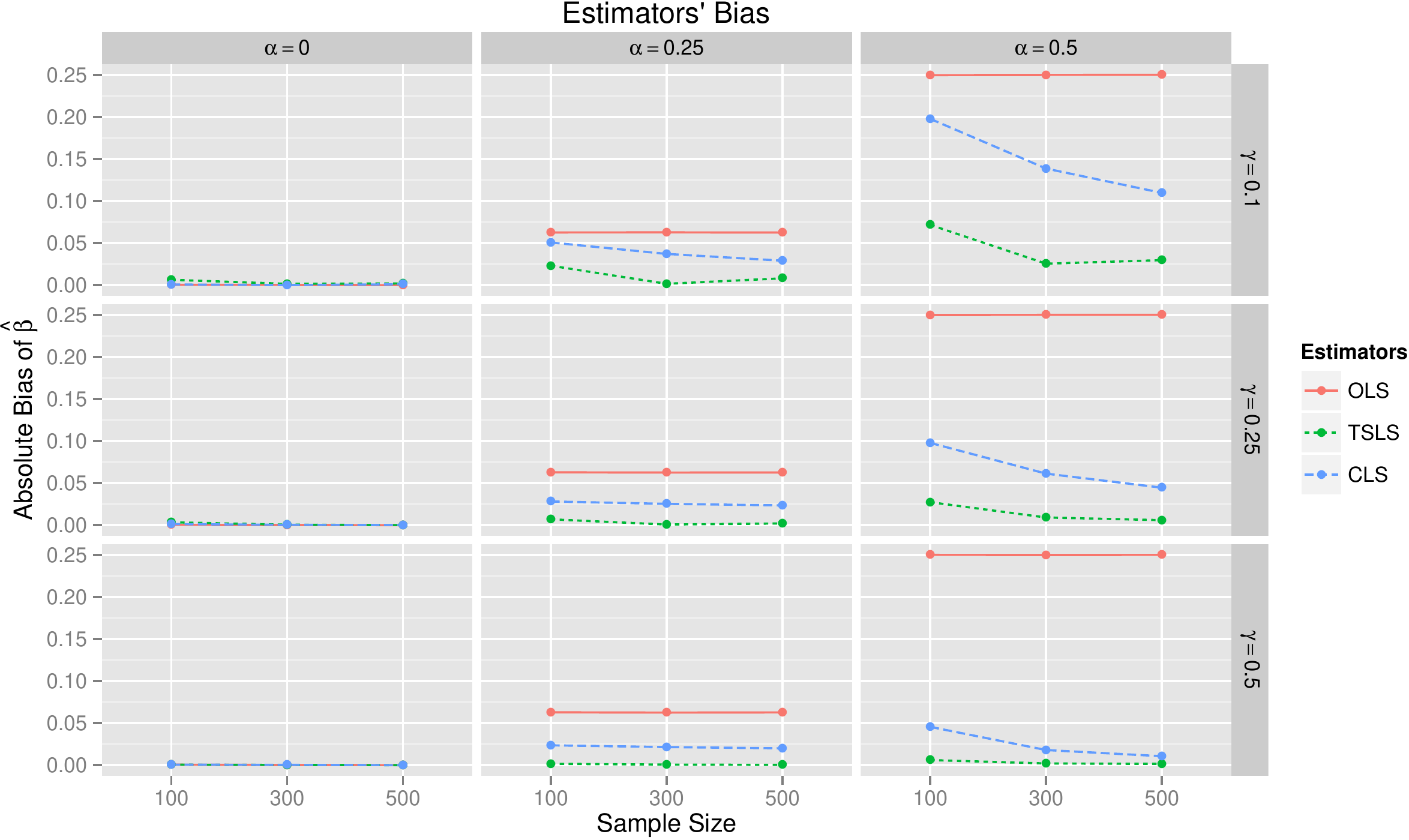}
   \caption{Monte Carlo estimates of the absolute value of the bias of
     the three estimators of interest, under the simulation scenarios
     described in figure \ref{fig:beta}. Observe that the three estimators
     exhibit no bias, when no confounding is present. 
     That is, the OLS estimator exhibits less bias, when
     $\al=\Cor(X_{i},U_{i})$ is low. Also, note that the finite-sample
     bias of the TSLS estimator tends to diminish with large sample
     sizes. This behavior is especially visible for large $\al$'s. 
     \label{fig:bias}}
\end{figure}
\begin{figure}[t]
  \centering
  \includegraphics[width=13cm]{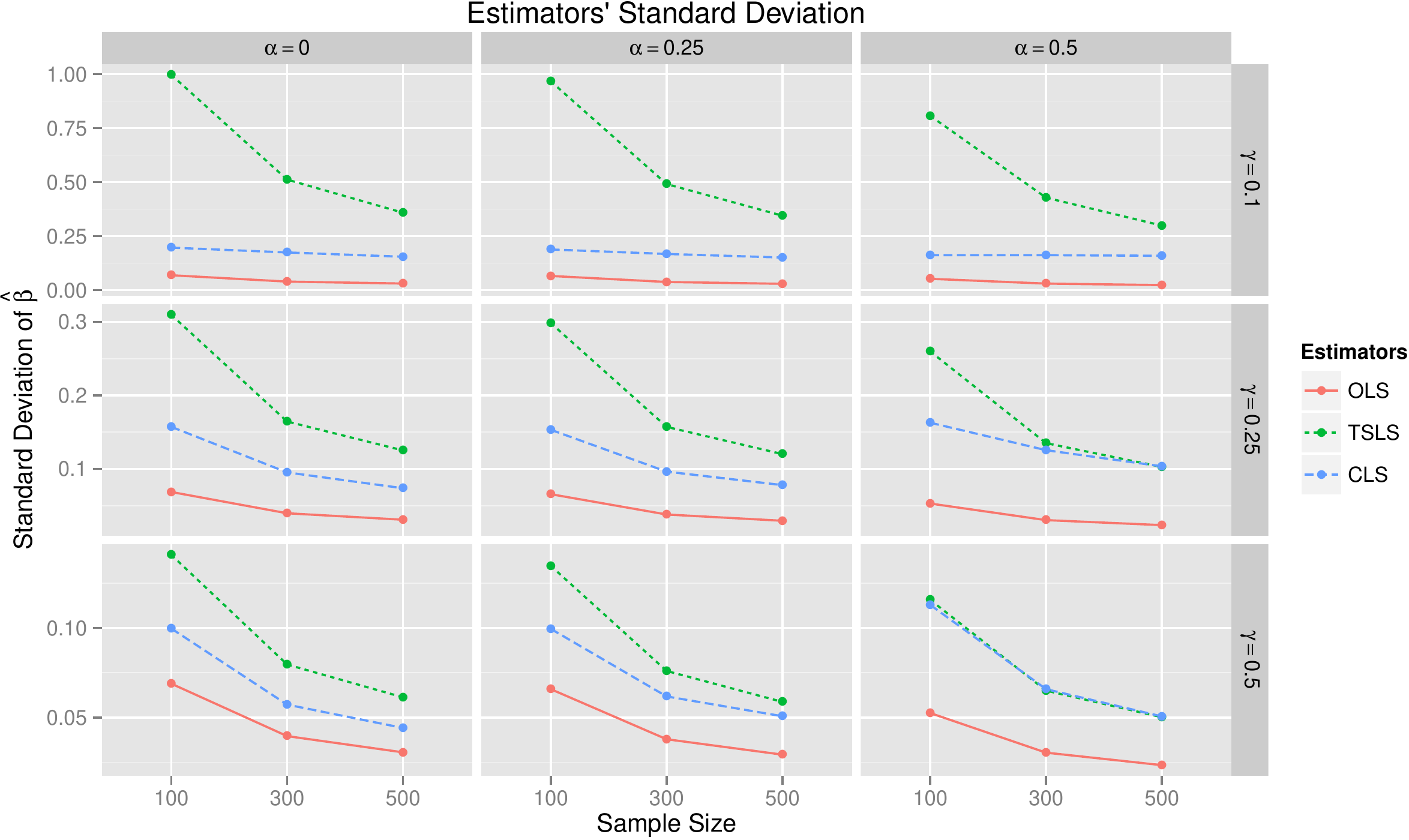}
   \caption{Monte Carlo standard deviation of the three estimators
     under scrutiny, under the scenarios
     described in figure \ref{fig:beta}. By corollary
     \ref{cor:bias-variance}, the variance of the OLS estimator is
     always smaller than its competitors, as verified in these
     simulations. Moreover, observe that the variance of the TSLS estimator increases
     as the strength of the instrument, $\ga=\Cor(X_{i},Z_{i})$, decreases.
     \label{fig:variance}}
\end{figure}
\begin{figure}[t]
  \centering
  \includegraphics[width=13cm]{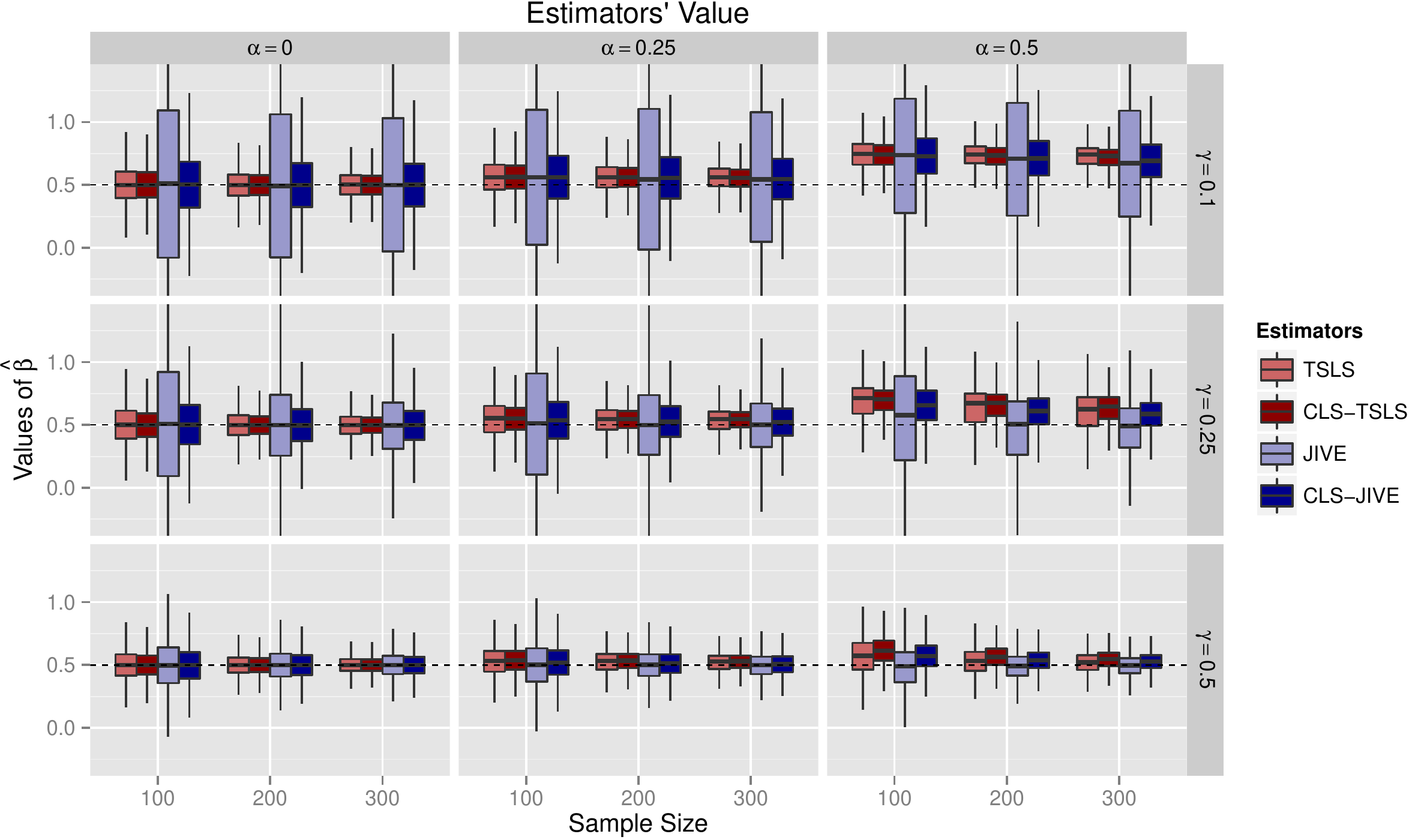}
   \caption{Comparison of the TSLS and JIVE estimators with their CLS
     counterparts, using $l=10$ uncorrelated
     instruments, whose multiple correlation with the outcome is given
     by $\ga$, and with $\al$ measuring the strength of bias, as in
     Model I. Data have been simulated using Model
     II from section \ref{sec:sim model ii}. The bootstrap estimate of
     $\pi_{n}$ for the CLS-JIVE is based on $B=100$
     resamples, as described in section \ref{sec:bootstrap} on
     bootstrap CLS estimation. All scenarios have
     been repeated over $10^{5}$ Monte Carlo iterations. 
     \label{fig:bcls1}}
\end{figure}
\begin{figure}[t]
  \centering
  \includegraphics[width=13cm]{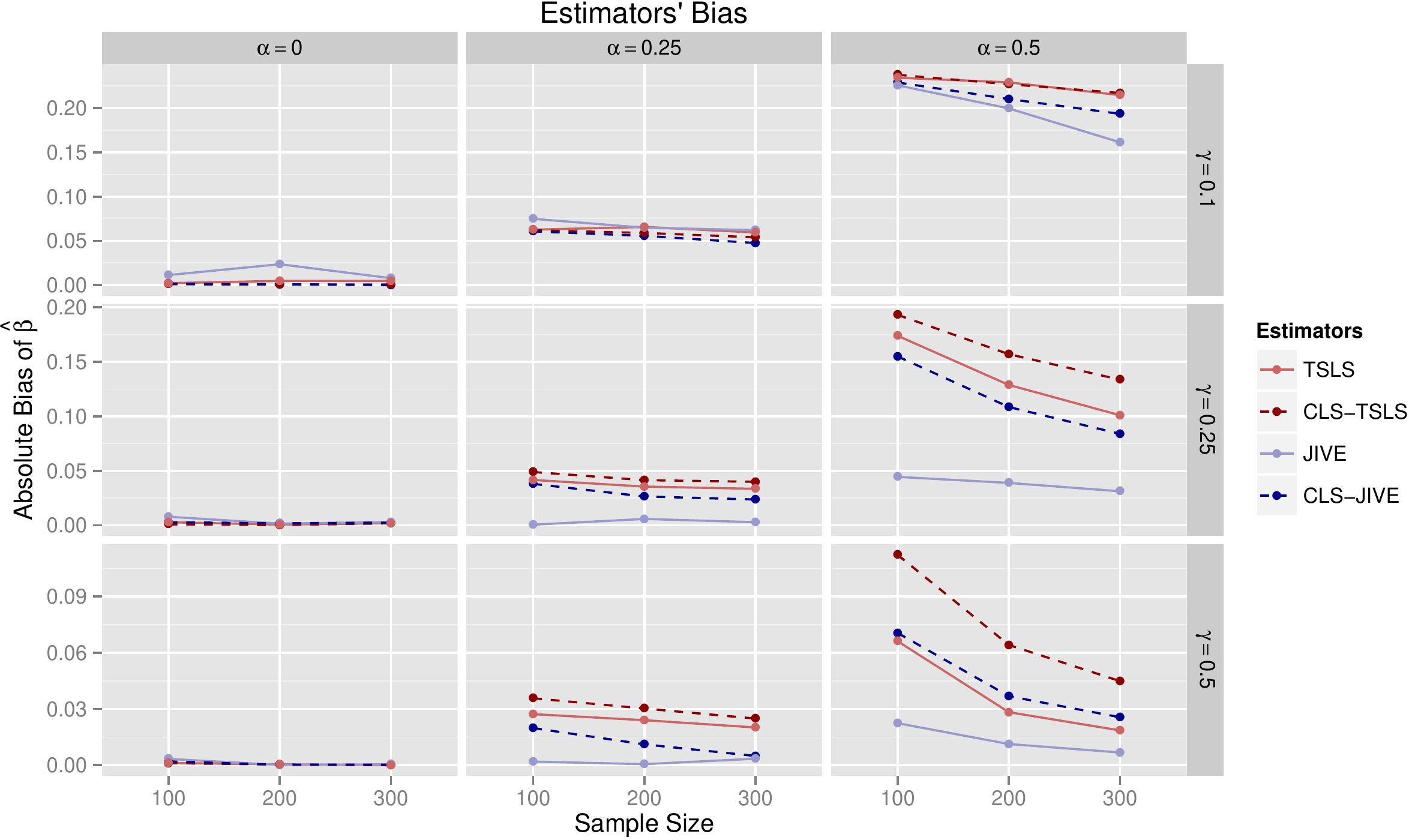}
   \caption{Monte Carlo estimates of the absolute value of the bias of
     four estimators of interest for Model II with $l=10$ uncorrelated
     instruments, whose multiple correlation with the outcome is given
     by $\ga$, and with the strength of the bias being denoted by
     $\al$. As in figure \ref{fig:bcls1}, all scenarios have been
     repeated over $10^{5}$ Monte Carlo iterations. 
     \label{fig:bcls2}}
\end{figure}

\subsection{Simulation Model II}\label{sec:sim model ii}
    We extend Model I to the case of several valid instruments. For
    convenience, these instruments are assumed to be uncorrelated. The
    second-level equation is taken to be identical to the second-level
    equation in equation (\ref{eq:model i}). The first-level equation,
    by contrast, now includes a vector of instruments, such that 
    \begin{equation}\label{eq:model ii}
        x_{i} = \bz_{i}\bla + u_{i}\al + \de_{i},
    \end{equation}
    for every $i=1,\ldots,n$; and where
    $\bz_{i}:=(z_{i1},\ldots,z_{il})$ is a row vector of $l$
    uncorrelated instruments. The strengths of each of these instruments are
    controlled by a column vector of parameters denoted by 
    $\bla:=(\la_{1},\ldots,\la_{l})\pri$. We here assume that
    the $\la_{j}$'s are held constant such that $\la_{j}:=\la$, for
    every $j=1,\ldots,l$.
    As for Model I, we fix $\var(Y_{i})=\var(X_{i})=1$, and generate
    the $Z_{ij}$'s and the $U_{i}$'s from a standard normal
    distribution, with $\sig_{z}^{2}:=1$ and $\sig^{2}_{u}:=1$,
    respectively. The formula for the error variance of the
    second-level equation, $\sig^{2}_{\ep}(\al)$, is identical to the
    one used in Model I. 

    The error variance for the first-level equation, denoted by
    $\sig^{2}_{\de}$, is also controlled by a parameter
    $\ga$, which is here defined to be the coefficient of multiple correlation
    between each $X_{i}$ and the $i\tth$ vector of $Z_{ij}$'s. As before, using the Bienaym\'e
    formula, the variance of each $X_{i}$ can be expanded as follows, 
    \begin{equation}\notag
        \var(X_{i}) = \sum_{j=1}^{l}\la^{2}_{j}\var(Z_{ij}) +
        \al^{2}\var(U_{i}) + \var(\de_{i}),
    \end{equation}
    which simplifies to $\var(X_{i})=l\la^{2}+\al^{2} + \sig^{2}_{\de}$,
    by our choice of variances for the $Z_{ij}$'s and $U_{i}$'s. When
    specifying $\var(X_{i})=1$, this gives
    $\sig^{2}_{\de}(\al,\ga,l)=1-(l\la^{2}+\al^{2})$, and moreover, when enforcing the
    positiveness of $\sig^{2}_{\de}$, we obtain
    $\la<\sqrt{(1-\al^{2})/l}$. Next, if we choose
    $\la:=\ga/\sqrt{l}$, then the parameter $\ga$ can be seen to correspond
    to the multiple correlation coefficient between each $X_{i}$ and
    the $i\tth$ vector of $Z_{ij}$'s. Indeed, we have $\ga^{2}=\br^{T}\bR^{-1}\br$, in which 
    $\br:=(r_{xz},\ldots,r_{xz})\pri$, with $r_{xz}:=\Cor(X_{i},Z_{ij})=\ga$;
    and where $\bR$ is the correlation matrix of the $i\tth$ vector of $Z_{ij}$'s, such
    that $R_{ab}:=\Cor(Z_{ia},Z_{ib})$, for every $a,b=1,\ldots,l$. Thus,
    as in Model I, we again obtain the upper bound,
    $\ga<\sqrt{1-\al^{2}}$, as well as
    $\sig_{\de}^{2}(\al,\ga)=1-(\ga^{2}+\al^{2})$.

\subsection{Monte Carlo Summary Statistics}\label{sec:mc statistics}
We evaluate the finite-sample performance of the estimators of interest
by comparing the Monte Carlo estimates of three different
population statistics. For every candidate estimator,
$\bbe_{n}^{\dag}$, its Monte Carlo distribution is given by the following empirical
distribution function (EDF), $\wh{F}(b) :=
T^{-1}\sum_{t}\cI\{\bbe_{nt}^{\dag}\leq b\}$,
where $\cI\{f_{t}\}$ denotes the indicator function taking a value of
one, if $f_{t}$ is true, and zero otherwise. For each simulation scenario,
we draw $T:=10^{5}$ realizations from the two models described in sections
\ref{sec:sim model i} and \ref{sec:sim model ii}.

Using these simulated samples, we compute the Monte Carlo estimates of
the bias, variance, and MSE; denoted by $\bias_{\wh{F}}(\bbe_{n}^{\dag})$,
$\var_{\wh{F}}(\bbe_{n}^{\dag})$, and $\mse_{\wh{F}}(\bbe_{n}^{\dag})$,
respectively. In figure \ref{fig:beta}, we have reported the
Monte Carlo distribution of the three estimators of interest under $T$ simulations from 
model I; whereas in figures \ref{fig:mse}, \ref{fig:bias},
and \ref{fig:variance} we have reported the Monte Carlo MSE, 
squared bias, and variance, respectively. The
quantities in these three figures have been square-rooted in order to
facilitate the comparison of these statistics with the estimators' values
in figure \ref{fig:beta}. Similarly, in figure
\ref{fig:bcls1}, we have reported the Monte Carlo
distributions of the TSLS, JIVE, as well as their CLS counterparts;
with the Monte Carlo estimates of their absolute value bias being
described in figure \ref{fig:bcls2}. 

\subsection{Results for Model I (Single Instrument)}\label{sec:sim results i}
The behavior of the CLS was found to be mainly controlled by the
strength of the instrument, $Z$. When the instrument was strongly
correlated with the predictor $X$ --that is, for large values of
$\ga=\Cor(X_{i},Z_{i})$; the values of the CLS estimator were close
to the ones of the TSLS estimators, as can be observed in the last row
of figure \ref{fig:beta}. By contrast, when the instrument was weak --that is,
for small values of $\ga$, the values of the CLS estimator were closer
to the ones of the OLS estimator, as can be seen in the first row of
figure \ref{fig:beta}. 

Proposition \ref{pro:mse optimal} stated that the MSEs of
the OLS and TSLS estimators are bounded below by the MSE of the CLS
estimator when the true proportion $\pi_{n}$ is known. These
Monte Carlo simulations appear to support a partial analog of
this result when $\pi_{n}$ is evaluated from the
data. Indeed, on one hand, figure \ref{fig:mse} shows that the MSE of the OLS
estimator tends to be smaller than the MSE of the CLS estimator,
when no confounding is present; thereby showing that proposition
\ref{pro:mse optimal} does not strictly hold when $\pi_{n}$ is estimated from
the data. However, on the other hand, one can also observe from figure \ref{fig:mse} that the Monte
Carlo MSE of the CLS estimator is smaller than or equal to
the one of the TSLS estimator under all considered scenarios. Thus,
it seems that a weaker version of proposition \ref{pro:mse optimal}
may hold for estimated $\pi_{n}$, which would solely pertain to a
comparison between the behavior of the CLS and TSLS estimators. 
In particular, observe that for strong instruments (i.e.~for large values of
$\ga$), the CLS estimator behave as well as the TSLS estimator,
whereas for weak instruments (i.e.~small values of $\ga$), the CLS
estimator outperforms the TSLS estimator. 

The behavior of these estimators can be better understood
by separately considering their bias and variance.
In figures \ref{fig:bias} and \ref{fig:variance}, we have
respectively reported the Monte Carlo estimates of the bias and
variance of the OLS, TSLS and CLS estimators. Naturally, the bias of
the three estimators tends to increase with the strength of the
confounder, which is controlled by $\al=\Cor(X_{i},U_{i})$. 
In particular, the bias of the OLS estimator becomes larges
as $\al$ increases. By contrast, the bias of the TSLS estimator
remains small for every value of $\al$. In fact, as stated in
Corollary \ref{cor:bias-variance}(i), the bias of the OLS
estimator is bounded from below by the bias of the TSLS
estimator. Moreover, the finite-sample
bias of the TSLS estimator can also be observed to decrease as the sample
size increases. As predicted, the bias of the CLS estimator is
comprised between the ones of the two other estimators; and the bias
of the CLS estimator approaches the one of the TSLS estimator, as the strength
of the instrument increases. 

Figure \ref{fig:variance} describes the behavior of the variance of
the estimators of interest under our various simulation scenarios. 
The variance of the three estimators tends to decrease as the
sample size increases. This downward trend is especially
noticeable for the TSLS estimator, which exhibits a high level of
variability, when the instrument is weak (i.e.~for small values of $\ga$).
As predicted by Corollary \ref{cor:bias-variance}(ii), the variance of
the TSLS estimator can be observed to be bounded from
below by the variance of the OLS estimator. In the presence of 
weak instruments, the CLS estimator's variance is close to the one of the OLS
estimator. As $\ga$ increases, however, the variance of the CLS
estimator converges to the one of the TSLS estimator.

\subsection{Results for Model II (Multiple Instruments)}\label{sec:sim
  results ii}
Our second set of simulations aimed to assess whether the use of an
estimator possessing better finite-sample properties could
be incorporated into the CLS framework. In figure \ref{fig:bias}, we
have already seen that the TSLS estimator suffers from a substantial finite-sample
bias. This was found to be especially the case when the instruments of
interest are comparatively weak, and the
bias is large. In particular, previous authors have shown that the
TSLS estimator's bias tends to be especially large, when several
instruments are used \citep{Angrist1995}. This limitation of the TSLS
estimator has been addressed in the literature by the
introduction of the JIVE, which was described in section \ref{sec:jive}. This
second set of simulations is thus based on $l=10$ uncorrelated
instruments, and allow us to compare the relative merits of using
either the TSLS estimator or the JIVE within the CLS framework. 
Consequently, we will refer to using the TSLS and using the JIVE
as unbiased estimators within the CLS, as the CLS-TSLS and CLS-JIVE
estimators, respectively. 

As predicted, the JIVE performs better than the TSLS estimator, when
used in conjunction with strong instruments, and in the presence of a
substantial amount of confounding (i.e.~$\al\geq0.25$), as can be seen from figures
\ref{fig:bcls1} and \ref{fig:bcls2}. Note, however, that for weak
instruments (i.e.~when the multiple correlation coefficient is $\ga=0.1$), the JIVE's
variance is very large. The TSLS estimator should therefore 
be favored under these scenarios, if one's choice of estimator is
motivated by a desire to minimize the MSE. 

The benefits of using the JIVE translate into corresponding
improvements when using the CLS-JIVE. This relationship is especially visible when
considering the bottom right panel of figure \ref{fig:bcls2}. Under a
set of strong instruments (i.e.~with a large multiple correlation
coefficient $\ga$), and under a substantial amount of confounding
(i.e.~large $\al$); one can
observe that the JIVE has a smaller finite-sample bias than the
TSLS estimator. Similarly, under the same scenario, the CLS-JIVE has
a correspondingly smaller finite-sample bias than the CLS-TSLS
estimator. This improvement in the CLS-JIVE was particularly
remarkable, because the proportion parameter,
$\wh\pi_{n}$, was estimated using only $B=100$ bootstrap samples. 
Thus, it appears that a relatively small number
of resampling is sufficient to produce a CLS-JIVE estimator that
outperforms the CLS-TSLS estimator. One may therefore conjecture that the CLS
framework could be used in conjunction with other
asymptotically unbiased estimators, even when the proportion
parameter is not available in closed-form.

\section{Applications to Econometrics}\label{sec:real}
Our proposed methods have been applied to a re-analysis of a classic
data set in econometrics, originally published by
\citet{Angrist1991}, which aimed to relate educational attainment
with earnings. This particular study has been the subject of
numerous replications and re-analysis, and therefore provides us with a
well-known example for evaluating the performance of CLS
estimation in a real-world data set. 

\subsection{Quarter-of-Birth as Instrument}\label{sec:real intro}
\citet{Angrist1991} reported a small but persistent seasonal pattern
of educational attainment over several decades between the 1920s and
the 1950s. They observed that two discrepant regulations in the
United States during that period have led to a `natural experiment', in
which individual differences in completed years of education could be
predicted by an individual's season of birth. On one hand, nationwide
school-entry requirements controlled the age at which a given child
began school. Indeed, at that period in the US, all children were expected to
reach six years of age by the first of January of their first year at
school. Thus, children born early in the year were likely to be older
than their peers in the same class. On the other hand, state-specific
compulsory schooling laws solely required pupils to remain in school
until their sixteenth birthday. Therefore, pupils born in the first
quarters of the year, wishing to leave school early, could do so at an
earlier stage than their peers born later in the year. 

These two regulations --school-entry requirements, and compulsory
schooling laws-- therefore conspired to enable children born in the
early quarters of the year to complete a smaller number of years of
education, if they were so inclined. Crucially, \citet{Angrist1991}
highlights that the randomness of an individual's birth date
is unlikely to be related to other events in an
individual's life; thereby precluding quarter-of-birth from being a
significant predictor of an individual's revenue later in life. Thus,
quarter-of-birth could be argued to constitute a legitimate
instrument for education attainment, fulfilling the exclusion
criterion, in the sense that it is not directly related to
earnings. Note, however, that some authors have disputed the validity
of quarter-of-birth as an instrument for education \citep{Bound1996}.
\begin{figure}[t]
\centering\small
\begin{tikzpicture}
    \draw (-3,0) node[draw,inner sep=8pt](z){$Z_{1}$};
    \draw (0,0) node[draw,inner sep=8pt](x){$X_{1}$};
    \draw (+4,0) node[draw,inner sep=8pt](y){$Y$};
    \draw (2,-1.75) node[draw,inner sep=8pt](x2){$X_{2}$};
    \draw (+0,+1.5) node[draw,circle,inner sep=5pt](d){$\de$};
    \draw (+4,+1.5) node[draw,circle,inner sep=5pt](e){$\ep$};
    \draw (2,+1.75) node[draw,circle,inner sep=5pt](u){$U$};
    \draw[thick,->] (z) -- (x) node[midway,above]{$\bGa_{1}$};
    \draw[thick,->] (x) -- (y) node[midway,above]{$\bbe_{1}$};
    \draw[thick,->] (x2) -- (x) node[midway,anchor=north east]{$\bGa_{2}$};
    \draw[thick,->] (x2) -- (y) node[midway,anchor=north west]{$\bbe_{2}$};
    \draw[thick,->] (u) -- (x) node[midway,anchor=south east]{$\al$};
    \draw[thick,->] (u) -- (y) node[midway,anchor=south west]{$\al$};
    \draw[thick,->] (e) -- (y) node[midway,anchor=north west]{};
    \draw[thick,->] (d) -- (x) node[midway,anchor=north east]{};
\end{tikzpicture}
\caption{Graphical representation of the IV model described in
  equations (\ref{eq:angrist1}) and (\ref{eq:angrist2}), composed of a 
  vector of endogenous variables, $X_{1}$, and a vector of exogenous
  variables $X_{2}$. This graph corresponds to
  a two-level system of equations composed of
  $Y=X_{1}\bbe_{1}+X_{2}\bbe_{2}+U\al+\ep$, and 
  $X_{1}=Z_{1}\bGa_{1}+X_{2}\bGa_{2}+U\al+\de$.    
  This model can be seen to be a generalization of the simpler model
  described in figure \ref{fig:model}.}
  \label{fig:angrist}
\end{figure}
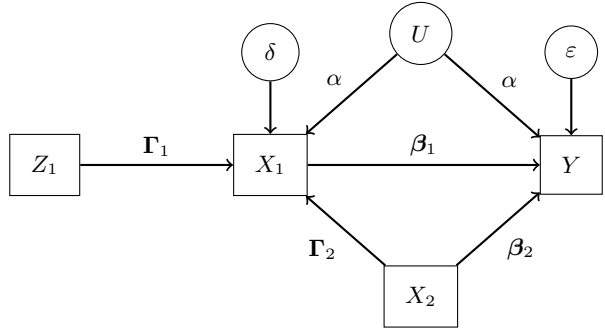
\begin{table*}[t]
\centering
\caption{Replications and extensions of the IV analysis in \citet{Angrist1991}.}
\label{tab:angrist}
\begin{3table}
\begin{tabular}{@{}>{\centering}m{10pt}
                >{\centering}m{45pt}|>{\centering}m{35pt}>{\centering}m{35pt}
                >{\centering}m{35pt}>{\centering}m{40pt}m{1pt}@{}}
   \toprule
   \mcol{2}{c}{\textit{Covariates ($X_{2}$)}} & OLS & TSLS & JIVE & CLS-TSLS &\\
   \midrule
    &                & \\
   \mrow{3}{*}{\textbf{A.\tnote{a}}}  
    &\it Estimator:  &   0.0802   &  0.0769  &  0.0755  & 0.0800 &\\
    &\it Std.~Error: &  (0.0004)  & (0.0150) & (0.0210) & (0.0126)  & \\
    &\it Proportion: &    --      &   --     &          &  (0.95)   & \\
    &                & \\
   \mrow{3}{*}{\textbf{B.\tnote{b}}}   
    &\it Estimator:  &   0.0802   &  0.1398  & -0.1276  &  0.1254 &\\
    &\it Std.~Error: &  (0.0004)  & (0.0334) & (1.7233) & (0.0317) & \\
    &\it Proportion: &    --      &   --     &          &  (0.24)  & \\
    &                & \\
   \mrow{3}{*}{\textbf{C.\tnote{c}}}   
   &\it Estimator:   &   0.0701   &  0.0669  &  0.0650  &  0.0700 &\\
   &\it Std.~Error:  &  (0.0004)  & (0.0151) & (0.0234) & (0.0097) & \\
   &\it Proportion:  &     --     &    --    &          &  (0.96)  & \\
    &                & \\
   \mrow{3}{*}{\textbf{D.\tnote{d}}}   
   &\it Estimator:   &   0.0701   &  0.1065  &  0.0224  &  0.0899 &\\
   &\it Std.~Error:  &  (0.0004)  & (0.0334) & (2.1380) & (0.0235) & \\
   &\it Proportion:  &    --      &   --     &          &  (0.46)  & \\
   \bottomrule
\end{tabular}
\begin{tablenotes}
    \item[a] This only includes ten dummies for the years of birth.
    \item[b] This includes years of birth,
      and age with the exclusion of the 1929 dummy.
    \item[c] Covariates include years of birth,
      and some extraneous covariates described in the text. 
    \item[d] Covariates are years of birth, age,
      and other extraneous covariates, with the exclusion of the 1929 dummy.
\end{tablenotes}
\end{3table}
\end{table*}

\subsection{Model with Extraneous Covariates}\label{sec:real model}
The model used by \citet{Angrist1991} generalizes the IV model
described in sections \ref{sec:ols} and \ref{sec:tsls}. Here, the
vector of $k$ predictors, $X$, is partitioned into $k_{1}$ endogenous
variables denoted by $X_{1}$, and $k_{2}$ exogenous variables denoted
by $X_{2}$. In Angrist and Krueger's model, the sole endogenous variable of
interest is the completed years of education of each
subject. Therefore, we have $k_{1}=1$. The outcome variable,
$Y$, which represents the log-transformed weekly wage in dollars of
each subject, is then modelled as follows, 
\begin{equation}\label{eq:angrist1}
      Y = X_{1}\bbe_{1} + X_{2}\bbe_{2} + \ep,
\end{equation}
where $\bbe_{1}$ and $\bbe_{2}$ are column vectors of dimension
$k_{1}$ and $k_{2}$, respectively. The endogeneity of $X_{1}$ leads to
the use of a row vector of instruments, $Z_{1}$, of dimension $1\times
l_{1}$, here denoting a set of dummy variables for the interactions
between quarters and years of birth. These
instruments are combined with the $k_{2}$ exogenous variables from the
second-level equation in order to produce the following first-level
equation,
\begin{equation}\label{eq:angrist2}
      X = Z_{1}\bGa_{1} + X_{2}\bGa_{2} + \de;
\end{equation}
where $\bGa_{1}$ and $\bGa_{2}$ are vectors of parameters of order
$l_{1}\times k$ and $k_{2}\times k$, respectively; where
$k:=k_{1}+k_{2}$. A graphical representation of this
model is given in figure \ref{fig:angrist}.

In matrix notation, given a sample of $n$ subjects, this model can be
expressed with respect to an $n$-dimensional column vector of
error terms, $\bep$, for the first-level equation; and a matrix, $\bD$,
of error terms of order $n\times k_{1}$ for the second-level
equation. Altogether, we thus have the following linear system, 
\begin{equation}\notag
  \begin{aligned}
      \by &= \bX_{1}\bbe_{1} + \bX_{2}\bbe_{2} + \bep,\\ 
      \bX &= \bZ_{1}\bGa_{1} + \bX_{2}\bGa_{2} + \bD.
  \end{aligned}
\end{equation}
Moreover, we can construct the following block matrices,
$\bX:=[\bX_{1}\;\bX_{2}]$ and $\bZ:=[\bZ_{1}\;\bX_{2}]$ that are of
order $n\times k$ and $n\times l$, respectively; in which we have used
$k:=k_{1}+k_{2}$ and $l:=l_{1}+k_{2}$. In addition, we also define the
vectors of parameters $\bbe:=[\bbe_{1}^{T}\;\bbe_{2}^{T}]^{T}$ and
$\bGa:=[\bGa_{1}^{T}\;\bGa_{2}^{T}]^{T}$, of order $k\times1$
and $l\times k$, respectively. Equipped with these block matrices, we
can immediately recover the standard TSLS estimator formula described in section
\ref{sec:tsls}. It also follows that this model is well-identified
whenever $l_{1}\geq k_{1}$ is satisfied, as in the model at hand.

\subsection{Results of the Re-analysis}\label{sec:real results}
The results described in table 4 of \citet{Angrist1991} have been
replicated and extended. In this portion of their analysis,
Angrist and Krueger have considered the cohort of men born
between 1920 and 1929. This constitutes a sample of $n=247,199$
subjects. All data are here based on the 1970 US census. Using the notation
introduced in equations (\ref{eq:angrist1}) and (\ref{eq:angrist2}),
the outcome variable, $Y$, is defined as the mean log-transformed
weekly wages; the endogenous variable, $X_{1}$, is the number of
completed years of education; and the instrument, $Z_{1}$, is
composed of a vector of interaction terms between quarter-of-year dummies and
year-of-birth dummies, totalling 40 different instruments. 

In addition, the authors have also considered different sets of exogenous covariates,
denoted by $X_{2}$ in equations (\ref{eq:angrist1}) and
(\ref{eq:angrist2}). The choice of exogenous covariates has been
reported as covariate scenarios A to D in table \ref{tab:angrist}.
All scenarios include ten dummy variables for each year of birth, except
scenarios B and D, in which the 1929 dummy variable has been removed due to
multicollinearity, following \citet{Angrist1991}. In scenarios B
and D, this is supplemented by an age covariate. (Note that we have
not included age squared in this analysis as was conducted in
\citet{Angrist1991}, since age and age squared were found to be almost
perfected correlated.) Finally, scenarios C and D include some further
dummy variables for race, marital status, eight different regions of
residence, and whether or not the subjects were primarily located in a
standard metropolitan statistical area (SMSA).

For scenarios A and C, the OLS and TSLS columns in table
\ref{tab:angrist} are exact replicates of the results described in 
\citet{Angrist1991}. The values and standard errors for these
estimators are slightly different under scenarios B and D, due to the
non-inclusion of age squared in the present analysis. 
The variance of the CLS estimator was computed using the bootstrap, as
described in section \ref{sec:bootstrap variance}. As expected, one
can observe that the CLS strikes a balance between the OLS and the TSLS
estimators, such that for all four scenarios, the value of the CLS is
comprised between the one of the OLS and the one of the TSLS. The
value taken by the estimate of the proportion
parameter has also been reported. By comparing scenarios A and C
with scenarios B and D in table \ref{tab:angrist}, it
can be seen that the inclusion (resp.~non-inclusion) of the age
variable leads to a decrease (resp.~increase) in the value of
$\wh\pi_{n}$. 

Thus, this re-analysis suggests that while the use of quarter-of-birth
as an IV for education may be justified when age is included as a
supplementary exogenous variable in the analysis; it appears that an
estimator closer to the OLS is sufficient, when the age variable is
not included. The JIVE and its CLS counterpart have also been reported
for comparison. The behavior of these two estimators is comparable to
the one of the TSLS and CLS-TSLS estimators. Note that for
computational convenience, the variances
of the CLS estimators and the proportion parameter of the CLS-JIVE
have been here estimated using bootstraps based on solely $10^2$
resamples. This real data analysis thus demonstrates that a small number of
bootstrap samples suffices to produce reasonable estimates of the
standard errors of the CLS estimators. 

\section{Conclusions}\label{sec:conclusion}
In this paper, we have shown that different IV and non-IV estimators
can be the object of convex combinations, and the proportion
parameters of these combinations can be consistently estimated from
the data. Such CLS estimators are therefore particularly attractive, 
since they automatically down-weight the influence of weak
instruments, when these are not expected to lead to a large reduction
in bias. Moreover, this inferential framework bears some
similarities with the Hausmann test. Theoretically, our 
proposed estimator minimizes an empirical MSE over a restricted class of
estimators, consisting of all the possible convex combinations of
the OLS and TSLS estimators. We have also seen the TSLS estimator in
the definition of the CLS can be replaced by other estimators such as
the JIVE. 

For finite $n$, the moments of the TSLS estimator and other $k$-class
estimators need not exist, as demonstrated by
\citet{Kinal1980}. It is common in such situations to assume that at least three
instruments are present. This condition ensures that the first two
moments of the estimators under scrutiny exist; and the first two
moments of the OLS and TSLS estimators are needed to compute our
proposed empirical estimator of the MSE. However, note that,
asymptotically, all moments exist and that, theoretically, this strategy can thus
be applied to any number of instruments. Indeed, irrespective of the
number of instruments used, every CLS estimator is guaranteed to be
asymptotically consistent. Observe that a similar issue arises for
all IV models, in the sense that such models are only
asymptotically identifiable. In the present case, our proposed
estimation procedure only possesses asymptotic first and second moments. However, just as
finite-sample non-identifiability is not a point of concern for the
use of IV methods in practice, the reliance of the CLS framework on
asymptotic moments does not constitute a significant hindrance to 
the general application of this method. 

The interpretation of the combination of several estimators such as the OLS and TSLS
estimators relies on the assumption of \textit{effect homogeneity}
[REFs needed]. That is, we are assuming that the causal effect is identical for all
subpopulations. This is a strong assumption, which needs not hold in
practice. Thus, further research will be needed to clarify the
assumptions required for employing the CLS, when effect
homogeneity is not expected to hold. Indeed, in such cases, the OLS and the
TSLS estimators may represent the local treatment effects in different
subpopulations. Therefore, the resulting convex combination of such
estimators may be difficult to interpret. Additional assumptions may
be required to ensure that the different estimators of interest are
sufficiently comparable to be combined in this fashion. 


Observe that the estimators utilized to produce the CLS estimator do not need to
share the same data. Indeed, when constructing a combination of the
OLS and TSLS estimators, only the TSLS estimator relies on the
instrument, $Z$. Thus, one may also consider how
such a framework could be extended to other modelling strategies,
such as mixed-effects models for longitudinal models
\citep[see][]{Wooldridge2002}. Similarly, this method could also be
extended to models including measurement errors. Calibrated regression
is often used in conjunction with explanatory variables, in order to
diminish the effect of measurement error. In such cases, the resulting
combination utilizes estimators based on distinct data
sets. In addition, recall that of the theoretical results derived in
this paper are relying one the assumption that the instruments under
scrutiny are valid, in the sense that (A4) is assumed to hold. The
consequences of relaxing this assumption are difficult to anticipate,
and further research should certain consider such situations, as has
been done by previous authors in the case of the TSLS estimator 
\citep{JacksonInpress}.

Thus far, we have combined the OLS estimator with either the TSLS
estimator or the JIVE. Observe, however, that we are not restricted to
choosing the OLS as a reference estimator. Within the bootstrap CLS
framework described in section \ref{sec:bootstrap}, one could also
choose to combine the TSLS and the JIVE, for instance. In general, any pair of
estimators could be the object of a convex combination. For such a
combination to be useful, it suffices that these estimators are
ordered in terms of bias and variance, as in the canonical case of the
OLS and TSLS estimators given in proposition \ref{pro:bias-variance}. 
Another natural theoretical extension of the current work would be to
derive a central limit theorem for the CLS estimator. This would allow
researchers to obtain approximate confidence intervals for the CLS
estimator, using normal asymptotic theory. Such a central
limit theorem would also enable the construction of adequate statistical
tests for evaluating whether or not the values of individual
parameters are statistically significant. Such extensions are not
expected to be too arduous, since under the assumptions stated in this
paper, the CLS estimator is consistent; and moreover, estimators such
as the TSLS or JIVE are known to be asymptotically normally
distributed under standard assumptions. 

\appendix
\section{Proofs of Propositions}\label{sec:proof}
\begin{proof}[Proof of Proposition \ref{pro:bias-variance}]
The proof of (i) immediately follows from the definition of the
empirical bias in equation (\ref{eq:bias}), which implies that the
empirical bias of the TSLS estimator is identically zero, for every
realization. The proof of (ii) can
be conducted in two steps. Firstly, one can show that for every
pair of matrices $\bX$ and $\wh{\bX}:=\bH_{z}\bX$, we have 
\begin{equation}\label{eq:bias-variance1}
         \bX\pri\bX \succeq
         \wh\bX{}\pri\wh\bX.
\end{equation}
Observe that we have the following equivalence due to the symmetry of
$\bH_{z}$, 
\begin{equation}\label{eq:xx1}
      \wh\bX{}\pri\bX = (\bH_{z}\bX)\pri\bX 
      = \bX\pri\bH_{z}\bX = \bX\pri\wh\bX. 
\end{equation}
Secondly, the inner product of $\wh\bX$ can also be simplified
using the idempotency of $\bH_{z}$, such that 
\begin{equation}\label{eq:xx2}
     \bX\pri\wh\bX = \bX\pri\bH_{z}\bX =
     \bX\pri\bH_{z}\bH_{z}\bX = \wh\bX{}\pri\wh\bX.
\end{equation}
Then, expanding the dot product of $\bX-\wh\bX$, and applying
equalities (\ref{eq:xx1}) and (\ref{eq:xx2}), we obtain
\begin{equation}\notag
    (\bX - \wh\bX)\pri(\bX - \wh\bX) 
    = \bX\pri\bX - 2\bX\pri\wh\bX + \wh\bX{}\pri\wh\bX 
    = \bX\pri\bX - \wh\bX{}\pri\wh\bX. 
\end{equation}
Observe that the dot product, $(\bX - \wh\bX)\pri(\bX -
\wh\bX)$, is a Gram matrix, and therefore it is necessarily positive
semi-definite. Consequently, this implies that 
\begin{equation}\notag
     \bX\pri\bX - \wh\bX{}\pri\wh\bX = (\bX - \wh\bX)\pri(\bX -
     \wh\bX) \succeq \bzero,
\end{equation}
and hence $\bX\pri\bX \succeq
\wh\bX{}\pri\wh\bX$, by the definition of the positive semidefinite order.

Next, observe that the estimates of the error variances under the
OLS and TSLS estimation procedures are defined as 
\begin{equation}\notag
  \begin{aligned}
    (n-k)\wh\sig^{2}_{n} &:= (\by -
    \bX\wh\bbe_{n})\pri(\by - \bX\wh\bbe_{n}) \\
    &\leq (\by - \bX\wti\bbe_{n})\pri(\by -
    \bX\wti\bbe_{n}) =: (n-k)\wti\sig^{2}_{n},
  \end{aligned}
\end{equation}
where the inequality follows from the optimality of the OLS; and therefore, 
\begin{equation}\notag
    \wh\sig^{2}_{n}(\bX\pri \bX)^{-1} \preceq\;
     \wti\sig^{2}_{n}(\wh{\bX}{}\pri\wh{\bX})^{-1},
\end{equation}
since assumption (A6) guarantees that both sides are invertible. 
\end{proof}

\begin{proof}[Proof of Corollary \ref{cor:bias-variance}]
   For inequality (i), observe that the empirical bias of the TSLS
   estimator is identically zero for every realization and every
   $n$. Moreover, by equation \ref{eq:bias}, the empirical bias of the OLS
   estimator is given by 
   \begin{equation}\notag
       \wh\bias{}^{2}(\wh\bbe_{n}) = (\wh\bbe_{n} -
       \wti\bbe_{n})(\wh\bbe_{n} - \wti\bbe_{n})\pri \succeq \bzero,
   \end{equation}
   since the LHS is a Gram matrix, and is therefore positive
   semidefinite for every realization. Since it holds for every $n$,
   this positive semidefiniteness is preserved in the limit. Moreover, the
   inequality in (ii) immediately follows from the
   fact that the variances of these two estimators converge to the
   zero matrix. 
\end{proof}

\begin{proof}[Proof of Proposition \ref{pro:pi}]
   The optimal value of $\pi_{n}$ can be found by minimizing the
   criterion of interest, which will be denoted by $f(\pi):=\mse(\bar\bbe_{n}(\pi))$.
   For expediency, we will expand this criterion as was done in
   equation (\ref{eq:mse pi}) such that 
   \begin{equation}\notag
       \tr f(\pi) = \tr(\pi^{2}M_{1} + 2\pi(1-\pi)C + (1-\pi)^{2}M_{2}),
   \end{equation}
   with $M_{1}:=\mse(\wh\bbe_{n})$, $C:=\cse(\wh\bbe_{n},\wti\bbe_{n})$, and
   $M_{2}:=\mse(\wti\bbe_{n})$; and where recall that $\wh\bbe_{n}$ and
   $\wti\bbe_{n}$ denote the OLS and TSLS estimators,
   respectively. Since the derivative is a linear operator, it
   commutes with the trace, and we obtain
   \begin{equation}\notag
        \tr(\partial f/\partial\pi) = 
        2\pi\tr(M_{2} - 2C + M_{1}) - 2\tr(M_{2} - C),
   \end{equation}
   Setting this expression to zero, yields $\pi_{n} :=
   \tr(M_{2}-C)/\tr(M_{2}-2C + M_{1})$, as required. Naturally,
   this minimization holds for every choice of $n$, thereby
   proving the first part of proposition \ref{pro:pi}. 

   In addition, one can show that this minimizer is unique by performing
   a second derivative test, such that we obtain 
   \begin{equation}\label{eq:second derivative}
        \tr(\partial^{2}f/\partial\pi^{2}) = 2\tr(M_{1} - 2C + M_{2}).
   \end{equation}
   Since by assumption, the random vectors, $\wh\bbe_{n}$ and $\wti\bbe_{n}$, are
   elementwise squared-integrable, the components, $\E[(\wh\be_{nj}
   - \be_{j})^{2}]$, of $M_{1}$ are finite for every
   $j=1,\ldots,k$. Hence, using the linearity
   of the trace, the MSE of $\wh\bbe_{n}$ can be treated as a sum of
   real numbers, thereby yielding the $L^{2}$-norm on $\R^{k}$, such
   that 
   \begin{equation}\notag
      \big(\tr\mse(\wh\bbe_{n})\big)^{1/2}
         = \bigg(\sum_{j=1}^{k}\E\!\lt[(\wh\be_{nj} - \be_{j})^{2}\rt]\bigg)^{1/2}
         =: ||\wh\bbe_{n}-\bbe||.
   \end{equation}
   The latter quantity will be referred to as the root (trace) MSE of
   $\wh\bbe_{n}$, and will be denoted by RMSE.

   By the same reasoning, it can be shown that $C$ and
   $M_{2}$ corresponds to the inner product, $\lan\wh\bbe_{n}-\bbe,\wti\bbe_{n}-\bbe\ran$, and
   the squared norm, $||\wti\bbe_{n}-\bbe||^{2}$, on $\R^{k}$. Thus,
   equation (\ref{eq:second derivative}) can now be re-expressed as follows,
   \begin{equation}\notag
      \tr(\partial^{2}f/\partial\pi^{2})
        = 2\Big(||\wh\bbe_{n}-\bbe||^{2}
        - 2\lan\wh\bbe_{n}-\bbe,\wti\bbe_{n}-\bbe\ran
        + ||\wti\bbe_{n}-\bbe||^{2}\Big).
   \end{equation}
   The Cauchy-Schwarz inequality can here be invoked to produce an
   upper bound on the cross-term in the latter equation, 
   \begin{equation}\notag
       \lan\wh\bbe_{n}-\bbe,\wti\bbe_{n}-\bbe\ran
       \leq ||\wh\bbe_{n}-\bbe||\cdot||\wti\bbe_{n}-\bbe||. 
   \end{equation}
   It then suffices to complete the square in order to obtain the following lower
   bound,
   \begin{equation}\notag
      \tr(\partial^{2}f/\partial\pi^{2})
      \geq 2\Big(||\wh\bbe_{n}-\bbe|| - ||\wti\bbe_{n}-\bbe||\Big)^{2} \geq 0,
   \end{equation}
   for every $n$, and where equality only holds when the RMSEs of
   $\wh\bbe_{n}$ and $\wti\bbe_{n}$ are identical, as required. 
\end{proof}

\begin{proof}[Proof of Proposition \ref{pro:mse}]
   Firstly, observe that the stochastic convergence of
   $\bar\bbe_{n}(\pi)$ to $\pi\wh\bbe + (1-\pi)\wti\bbe$ is immediate
   from the convergence of $\wh\bbe_{n}$ and $\wti\bbe_{n}$ to their
   respective limits, $\wh\bbe$ and $\wti\bbe$. Naturally, this
   holds for every $\pi\in[0,1]$. 
   
   Secondly, recall that the empirical MSE of $\bar\bbe_{n}(\pi)$ can
   be decomposed into a variance and a bias term, as in equation
   (\ref{eq:empirical mse}), such that for every $\pi$, we have
   \begin{equation}\notag
       \wh\mse(\bar\bbe_{n}(\pi)) 
       = \wh\var(\bar\bbe_{n}(\pi)) + \wh\bias{}^{2}(\bar\bbe_{n}(\pi)).
   \end{equation}
   The variances of both the OLS and TSLS estimators are known to
   converge to zero. Moreover, the bias of the TSLS estimator is also known to
   converge to zero, and this can be seen to be also true for the
   cross-bias term. Thus, the stochastic limit of
   $\wh\mse(\bar\bbe_{n}(\pi))$ reduces to the weighted limit of the empirical bias of
   the OLS estimator, $\pi^{2}\wh\bias{}^{2}(\wh\bbe_{n})$. Using the
   consistency of the TSLS estimator, this latter term satisfies
   \begin{equation}\notag
       \plim_{n\to\infty} \wh\bias(\wh\bbe_{n}) 
       = \plim_{n\to\infty} (\wti\bbe_{n} - \wh\bbe_{n})
       = \bias(\wh\bbe) = (\bbe - \wh\bbe),
   \end{equation}
   since $\wh\bbe_{n}\stack{p}{\longrightarrow}\wh\bbe$, and 
   $\wh\bbe=\lim_{n}\E[\wh\bbe_{n}]$; which is the true bias of the
   OLS estimator. The result then follows by using the continuous mapping
   theorem. 
\end{proof}

\begin{proof}[Proof of Corollary \ref{cor:pi}]
   The minimization of the empirical MSE follows from the arguments used
   in the first part of proposition \ref{pro:pi}, and simplifying the
   closed-form formula for $\wh\pi_{n}$ using our adopted definitions
   for the empirical bias of the TSLS estimator.
\end{proof}

\begin{proof}[Proof of Proposition \ref{pro:mse optimal}]
   This inequality relates the theoretical MSEs of the CLS,
   OLS and TSLS estimators. This result follows from the convexity of
   the MSE with respect to its argument, $\bbe^{\dag}_{n}$; in which
   $\bbe^{\dag}_{n}$ represents any candidate estimator. The trace of the MSE
   can indeed be seen to be a convex quadratic form, $\bx\pri\bA\bx$,
   where $\bA$ is here the identity matrix and
   $\bx:=\bbe^{\dag}_{n}-\bbe$. That is, for every estimator,
   $\bbe^{\dag}_{n}$, let
   \begin{equation}\notag
       f(\bbe_{n}^{\dag}-\bbe) 
       := \tr\big((\bbe^{\dag}_{n}-\bbe)(\bbe^{\dag}_{n}-\bbe)\pri\big)
       = (\bbe^{\dag}_{n}-\bbe)\pri(\bbe^{\dag}_{n}-\bbe),
   \end{equation} 
   using the cyclic property of the trace, which shows that
   this quadratic form is convex. Thus, for every two estimators,
   $\wh\bbe_{n}$ and $\wti\bbe_{n}$, and every $\pi\in[0,1]$; using
   the fact that $\bbe$ can always be expressed as $\pi\bbe +
   (1-\pi)\bbe$, we have  
   \begin{equation}\notag     
     \begin{aligned}
       f\Big(\pi\wh\bbe_{n} + (1-\pi)\wti\bbe_{n} - \bbe\Big)
       & = f\Big(\pi(\wh\bbe_{n}-\bbe) + (1-\pi)(\wti\bbe_{n} -\bbe)\Big) \\
       & \leq \pi f\Big(\wh\bbe_{n}-\bbe\Big) + (1-\pi)f\Big(\wti\bbe_{n}-\bbe\Big).
     \end{aligned}
   \end{equation} 
   Since by definition, $\bar\bbe_{n}(\pi):=\pi\wh\bbe_{n} +
   (1-\pi)\wti\bbe_{n}$, and moreover $\tr(\mse(\bbe^{\dag}_{n})) =
   \E[f(\bbe_{n}^{\dag})]$ for every $\bbe^{\dag}_{n}$, it then
   follows that, using the linearity of the expectation, 
   \begin{equation}\notag      
      \tr\mse(\bar\bbe_{n}(\pi)) \leq  
      \pi\tr\mse(\wh\bbe_{n}) + (1-\pi)\tr\mse(\wti\bbe_{n}).
   \end{equation}
   Finally, observing that the set of convex combinations of the form, $\pi
   f(\wh\bbe_{n}) + (1-\pi)f(\wti\bbe_{n})$, necessarily includes the
   endpoints of the corresponding line segment, we hence obtain 
   \begin{equation}\notag
     \begin{aligned}
       \min_{\pi}\Big(\tr\mse(\bar\bbe_{n}(\pi))\Big) 
          &\leq \min_{\pi}\lt(\pi\tr\mse(\wh\bbe_{n}) + (1-\pi)\tr\mse(\wti\bbe_{n})\rt) \\
          &\leq \min \Big\lb \tr\mse(\wh\bbe_{n}), \tr\mse(\wti\bbe_{n}) \Big\rb,
     \end{aligned}
   \end{equation}
   as required.
\end{proof}

\begin{proof}[Proof of Proposition \ref{pro:cls consistency}]
   We use a sandwich argument to prove that the CLS estimator is MSE
   consistent, as stated in (ii). For every $p\in[0,1]$, let 
   \begin{equation}\notag
       T_{n}(p) := 
        \tr\wh\mse(\bar\bbe_{n}(p))
        - \tr\mse(\bar\bbe_{n}(p)).
   \end{equation}  
   By proposition \ref{pro:mse}, we know that
   $T_{n}(p)\stack{p}{\longrightarrow}0$, for every $p$. However, the
   quantity of interest in the case of proposition \ref{pro:cls
     consistency} can be defined as follows, 
   \begin{equation}\notag
       T\as_{n}(p) := 
        \tr\wh\mse(\bar\bbe_{n}(p))
        - \tr\mse(\bar\bbe_{n}(\pi));
   \end{equation}
   with $p$ being chosen as $\wh\pi_{n}$, and where $\pi$ denotes
   the true proportion minimizing the theoretical MSE. 
   Firstly, observe that we can find a lower bound for
   $T\as_{n}(\wh\pi_{n})$ with respect to $T_{n}(p)$ by a judicious choice of
   $p$. That is, 
   \begin{equation}\notag
     \begin{aligned}
       T_{n}(\wh\pi_{n}) 
       & = \tr\wh\mse(\bar\bbe_{n}(\wh\pi_{n})) - \tr\mse(\bar\bbe_{n}(\wh\pi_{n})) \\
       & \leq \tr\wh\mse(\bar\bbe_{n}(\wh\pi_{n})) - \tr\mse(\bar\bbe_{n}(\pi)) 
       = T\as_{n}(\wh\pi_{n}),
     \end{aligned}
   \end{equation}
   since, by definition, $\pi:=\min_{p}\tr\mse(\bar\bbe(p))$.
   Secondly, one can also derive an upper bound for
   $T\as_{n}(\wh\pi_{n})$, as follows, 
   \begin{equation}\notag
     \begin{aligned}
       T\as_{n}(\pi) 
       &= \tr\wh\mse(\bar\bbe_{n}(\wh\pi_{n})) - \tr\mse(\bar\bbe_{n}(\pi)) \\
       &\leq \tr\wh\mse(\bar\bbe_{n}(\pi))) - \tr\mse(\bar\bbe_{n}(\pi)) 
       = T_{n}(\pi),
     \end{aligned}
   \end{equation}
   since $\wh\pi_{n}:=\min_{p}\tr\wh\mse(\bar\bbe(p))$. Therefore, we
   obtain the following sandwich inequality, 
   \begin{equation}\notag
        T_{n}(\wh\pi_{n}) \leq T_{n}\as(\wh\pi_{n}) \leq T_{n}(\pi),
   \end{equation}
   which could be re-expressed as follows, 
   \begin{equation}\notag
      |T_{n}\as(\wh\pi_{n})| \leq
      \max\{|T_{n}(\wh\pi_{n})|,|T_{n}(\pi)|\}
      \stack{p}{\longrightarrow} 0,
   \end{equation}
   where the weak convergence to zero was stated in proposition
   \ref{pro:mse}. Thus, we have demonstrated that 
   \begin{equation}\notag
       \min_{p} \tr\wh\mse(\bar\bbe_{n}(p)) 
       \stack{p}{\longrightarrow}
       \min_{p} \tr\mse(\bar\bbe_{n}(p)).
   \end{equation}
   By proposition \ref{pro:mse optimal}, we can use the MSEs of the OLS and TSLS
   estimators as upper bound for the RHS of the latter equation such
   that 
   \begin{equation}\notag       
      \min_{\pi} \tr\mse(\bar\bbe_{n}(\pi)) \leq 
            \tr\min\lb\mse(\wh\bbe_{n}),\mse(\wti\bbe_{n})\rb 
            \stack{p}{\longrightarrow} 0,
   \end{equation}
   since we know that the TSLS is MSE consistent. Therefore, 
   $\bar\bbe_{n}(\wh\pi_{n})\stack{L^{2}}{\to}\bbe$, as required. 
   Moreover, the weak consistency of the CLS estimator stated in (i) is a direct
   consequence of its MSE consistency \citep[see, for
   example,][p.313]{Bain1992}.  
\end{proof}

\begin{proof}[Proof of Corollary \ref{cor:sample optimal}]
   The proof of this inequality is analogous to the proof of
   proposition \ref{pro:mse optimal}, in which the theoretical
   expectation with respect to $F$, is replaced with the bootstrap
   expectation with respect to $F\as$. Here, the true parameter is
   taken to be $\E\as[\bbe^{\dag}_{n}]$, where $\bbe^{\dag}_{n}$ is an
   asymptotically unbiased estimator of interest, and 
   $\bar\bbe_{n}(\pi):=\pi\wh\bbe_{n} + (1-\pi)\bbe^{\dag}_{n}$ is the
   BCLS estimator. Letting $f(\bx):=\bx\pri\bx$, and using the
   convexity of this quadratic form, we obtain 
   \begin{equation}\notag     
       f\Big(\bar\bbe_{n}(\pi) - \E\as[\bbe^{\dag}_{n}]\Big)
          \leq \pi f\Big(\wh\bbe_{n} - \E\as[\bbe^{\dag}_{n}]\Big) +
          (1-\pi)f\Big(\bbe^{\dag}_{n} - \E\as[\bbe^{\dag}_{n}]\Big).
   \end{equation} 
   for every $\pi$. Moreover, by definition, $\tr(\wh\mse{}\as(\bar\bbe_{n})) =
   \E\as[f(\bar\bbe_{n})]$. Thus, using linearity, this gives
   \begin{equation}\notag      
      \tr\wh\mse{}\as(\bar\bbe_{n}(\pi)) \leq  
      \pi\tr\wh\mse{}\as(\wh\bbe_{n}) + (1-\pi)\tr\wh\mse{}\as(\bbe^{\dag}_{n}).
   \end{equation}
   The required inequality then follows by minimizing both sides with
   respect to $\pi\in[0,1]$, and noticing that every set of convex
   combinations contains the endpoints of the corresponding line
   segment, as in proposition \ref{pro:mse optimal}.
\end{proof}

\footnotesize
\singlespacing
\bibliography{/home/cgineste/ref/bibtex/Statistics,%
             /home/cgineste/ref/bibtex/Neuroscience}
\bibliographystyle{oupced3}

\addcontentsline{toc}{section}{Index}
\end{document}